\newif\ifanon
\anonfalse

\documentclass[11pt]{article}
\pdfoutput=1

\usepackage[T1]{fontenc}
\usepackage[utf8]{inputenc}
\usepackage[american]{babel}

\usepackage[letterpaper,margin=1in]{geometry}

\usepackage{ambib}
\usepackage{amsmath}
\usepackage{amsthm}
\usepackage{amssymb}
\usepackage{graphicx}
\usepackage{xcolor}
\usepackage{amsfonts}
\usepackage{bm}
\usepackage{subcaption}
\usepackage{pifont}
\usepackage{dsfont}
\usepackage[linesnumbered]{algorithm2e}
\usepackage{framed}
\usepackage{csquotes}
\usepackage{enumitem}

\usepackage{amthm}

\usepackage{complexity}
\let\H\relax

\let\EE\relax

\usepackage{ammacros}

\usepackage{booktabs}
\definecolor{myblue}{HTML}{0088cc}
\definecolor{myorange}{HTML}{f26924}
\definecolor{mygreen}{HTML}{3ec636}

\usepackage{hyperref}
\usepackage{xcolor}
\definecolor{myblue}{HTML}{0088cc}
\hypersetup{
  colorlinks=true,
  linkcolor=black,
  citecolor=black,
  filecolor=black,
  urlcolor=myblue,
}

\usepackage{local-models}
\newderivedmodel{\clocal}{computable-\local}
\newderivedmodel{\ulocal}{uncomputable-\local}
\newderivedmodel{\Clocal}{Computable-\local}
\newderivedmodel{\Ulocal}{Uncomputable-\local}

\newcommand*{\id}{\mathtt{id}} 
\newcommand*{\ddeg}{\mathtt{deg}} 

\newcommand*{\iin}{\mathrm{in}}
\newcommand*{\out}{\mathrm{out}}
\newcommand*{\accept}{\mathsf{accept}}
\newcommand*{\reject}{\mathsf{reject}}

\renewcommand{\AA}{\mathcal{A}}
\newcommand{\BB}{\mathcal{B}}
\newcommand{\CC}{\mathcal{C}}

\newcommand{\EE}{\mathcal{E}}

\newcommand{\GG}{\mathcal{G}}

\newcommand{\MM}{\mathcal{M}}
\newcommand{\NN}{\mathcal{N}}
\newcommand{\OO}{\mathcal{O}}

\newcommand{\VV}{\mathcal{V}}

\newcommand*{\rootn}{\mathfrak{N}}
\DeclarePairedDelimiter\descr{\langle}{\rangle}
\newcommand{\Left}{\mathrm{L}}
\newcommand{\Right}{\mathrm{R}}
\newcommand{\chl}{\mathrm{Ch}_{\mathrm{L}}}
\newcommand{\chr}{\mathrm{Ch}_{\mathrm{R}}}
\newcommand{\ch}{\mathrm{Ch}}
\newcommand{\Parent}{\mathrm{P}}
\newcommand{\north}{\mathrm{N}}
\newcommand{\east}{\mathrm{E}}
\newcommand{\south}{\mathrm{S}}
\newcommand{\west}{\mathrm{W}}
\newcommand{\Vi}{\mathcal{V}_{\textit{input}}}
\newcommand{\Ei}{\mathcal{E}_{\textit{input}}}

\makeatletter
\DeclareMathOperator{\regularlog}{log}
\renewcommand{\log}{\protect\@ifstar{\regularlog^*}{\regularlog}}
\makeatother

\newcommand{\Vin}{\VV^{\mathrm{in}}}
\newcommand{\Ein}{\EE^{\mathrm{in}}}
\newcommand{\Vout}{\VV^{\mathrm{out}}}
\newcommand{\Eout}{\EE^{\mathrm{out}}}
\newcommand{\Lin}{L^{\mathrm{in}}}
\newcommand{\Lout}{L^{\mathrm{out}}}

\newcommand*{\bad}{\mathrm{bad}}
\newcommand*{\good}{\mathrm{good}}
\newcommand*{\tree}{\mathrm{tree}}
\newcommand*{\goodtree}{\mathrm{goodTree}}
\newcommand*{\badtree}{\mathrm{badTree}}
\newcommand*{\row}{\mathrm{row}}
\newcommand*{\goodrow}{\mathrm{goodRow}}
\newcommand*{\badrow}{\mathrm{badRow}}
\newcommand*{\grid}{\mathrm{grid}}
\newcommand*{\goodgrid}{\mathrm{goodGrid}}
\newcommand*{\badgrid}{\mathrm{badGrid}}
\newcommand*{\turing}{\mathrm{Turing}}
\newcommand*{\goodturing}{\mathrm{goodTuring}}
\newcommand*{\badturing}{\mathrm{badTuring}}
\newcommand*{\consensus}{\mathrm{consensus}}
\newcommand*{\next}{\mathrm{next}}

\newcommand{\labels}[3]{\ensuremath{\mathcal{#1}_{\mathrm{#2}}^{\mathrm{#3}}}}
\newcommand{\prob}[2]{\Pi^{\mathrm{#1}}_{\mathrm{#2}}}

\newcommand{\botTM}{\bot_\mathrm{Turing}}

\usepackage{csquotes}
\usepackage{microtype}
\usepackage[capitalize,noabbrev]{cleveref}

\newenvironment{myabstract}%
{\list{}{\listparindent 1.5em
        \itemindent    \listparindent
        \leftmargin    1cm
        \rightmargin   1cm
        \parsep        0pt}%
    \item\relax}%
{\endlist}

\newenvironment{mycover}%
{\list{}{\listparindent 0pt
        \itemindent    \listparindent
        \leftmargin    1cm
        \rightmargin   1cm
        \parsep        0pt}%
    \raggedright
    \item\relax}%
{\endlist}

\newcommand{\myaff}[1]{\,$\cdot$\, {\small #1}\par\medskip}

\begin{document}

\begin{mycover}
{\huge\bfseries Is a \textsf{LOCAL} algorithm computable? \par}

\bigskip
\bigskip

\ifanon
\textbf{Anonymous authors}

\else
\textbf{Antonio Cruciani}
\myaff{Aalto University, Finland}

\textbf{Avinandan Das}
\myaff{Aalto University, Finland}

\textbf{Massimo Equi}
\myaff{Aalto University, Finland}

\textbf{Henrik Lievonen}
\myaff{Aalto University, Finland}

\textbf{Diep Luong-Le}
\myaff{Columbia University, United States}

\textbf{Augusto Modanese}
\myaff{CISPA Helmholtz Center for Information Security, Germany}

\textbf{Jukka Suomela}
\myaff{Aalto University, Finland}

\fi 

\bigskip
\end{mycover}

\begin{myabstract}
\noindent
Common definitions of the ``standard'' \textsf{LOCAL} model tend to be sloppy and even self-contradictory on one point: do the nodes update their state using an \emph{arbitrary} function or a \emph{computable} function? So far, this distinction has been safe to neglect, since problems where it matters seem contrived and quite different from e.g.\ typical local graph problems studied in this context.

We show that this question matters even for locally checkable labeling problems (LCLs), perhaps the most widely studied family of problems in the context of the \textsf{LOCAL} model. Furthermore, we show that assumptions about computability are directly connected to another aspect already recognized as highly relevant: whether we have any knowledge of $n$, the size of the graph. Concretely, we show that there is an LCL problem $\Pi$ with the following properties:
\begin{enumerate}
\item $\Pi$ can be solved in $O(\log n)$ rounds if the \textsf{LOCAL} model is uncomputable.
\item $\Pi$ can be solved in $O(\log n)$ rounds in the computable model if we know \emph{any} upper bound on~$n$.
\item $\Pi$ requires $\Omega(\sqrt{n})$ rounds in the computable model if we do not know anything about~$n$.
\end{enumerate}
We also show that the connection between computability and knowledge of $n$ holds in general: for any LCL problem $\Pi$, if you have any bound on $n$, then $\Pi$ has the same round complexity in the computable and uncomputable models.
\end{myabstract}

\thispagestyle{empty}
\setcounter{page}{0}
\clearpage

\section{Introduction}
\label{sec:introduction}

There is no single, universally-agreed-on definition of the ``standard'' \local model of distributed computing. And when a paper introduces the \local model, some aspects are often left unspecified, or the definition might even be self-contradictory.
A case in point is assumptions related to \textbf{computability}. By definition, the local state of a node after $T$ rounds is some function $f$ of its radius-$T$ neighborhood---but is this function a computable function, or an arbitrary function?

Many key papers state that one can use any function, without any restrictions related to computability \cite{linial-1992-locality-in-distributed-graph-algorithms,naor-stockmeyer-1995-what-can-be-computed-locally,ghaffari-kuhn-maus-2017-on-the-complexity-of-local}. At the same time, papers often use words such as ``processor'' \cite{linial-1992-locality-in-distributed-graph-algorithms,naor-stockmeyer-1995-what-can-be-computed-locally} to refer to individual nodes, as if they would be implementable as real-world computers that are restricted to computable functions. Some papers make it explicit that nodes are computationally bounded \cite{le-gall-nishimura-rosmanis-2019-quantum-advantage-for}.
Finally, there are some rare papers such as \cite{fraigniaud-goos-etal-2013-what-can-be-decided-locally} that explicitly acknowledge that one can define both the computable and uncomputable version of the model.

Many papers leave computability assumptions somewhat vague and open to interpretation, and when explicitly asked about this aspect, there is disagreement among the experts on whether the computable or uncomputable version is the ``right'' one. Yet what many of our colleagues seem to assume is that this aspect is not particularly relevant from the perspective of problems of interest to them.

The issue gets increasingly important (and confusing) with the emerging area of distributed quantum computing. When we define the quantum version of the \local model \cite{gavoille-kosowski-markiewicz-2009-what-can-be-observed,arfaoui-fraigniaud-2014-what-can-be-computed-without}, the natural formalization is something equivalent to quantum circuits or quantum Turing machines, which is necessarily computable---there is no quantum analogue of mapping radius-$T$ neighborhoods to local outputs. Combining this with a conventional uncomputable definition of the classical \local model would lead to paradoxical results: quantum-\local would be weaker than classical \local.

What we show in this work is that assumptions related to computability \emph{are} indeed critical even if we focus on problems and questions commonly studied in the context of the \local model.

\subsection{Locally Checkable Labelings}

Perhaps the most popular (and most restrictive) formal definition for a class of ``reasonable'' graph problems in this context is the family of \emph{locally checkable labeling problems}, or LCLs, first introduced by Naor and Stockmeyer in 1995 \cite{naor-stockmeyer-1995-what-can-be-computed-locally}. These are problems that can be specified by giving a \emph{finite} set of valid labeled local neighborhoods. There is a long line of research that has culminated in a near-complete classification of LCL problems in the \local model for general graphs, see e.g.\ \cite{%
    cole-vishkin-1986-deterministic-coin-tossing-with,
    naor-1991-a-lower-bound-on-probabilistic-algorithms-for,
    linial-1992-locality-in-distributed-graph-algorithms,
    naor-stockmeyer-1995-what-can-be-computed-locally,
    brandt-fischer-etal-2016-a-lower-bound-for-the,
    fischer-ghaffari-2017-sublogarithmic-distributed,
    ghaffari-harris-kuhn-2018-on-derandomizing-local,
    balliu-hirvonen-etal-2018-new-classes-of-distributed,
    chang-pettie-2019-a-time-hierarchy-theorem-for-the,
    chang-kopelowitz-pettie-2019-an-exponential-separation,
    rozhon-ghaffari-2020-polylogarithmic-time-deterministic,
    balliu-brandt-etal-2020-how-much-does-randomness-help,
    balliu-brandt-etal-2021-almost-global-problems-in-the,
    suomela-2020-landscape-of-locality-invited-talk,
    chang-2020-the-complexity-landscape-of-distributed,
    grunau-rozhon-brandt-2022-the-landscape-of-distributed%
}.

What makes LCL problems particularly nice and clean from the perspective of computability issues is that for any LCL problem $\Pi$, there are Turing machines that solve both of these tasks:
\begin{enumerate}[noitemsep]
    \item Given a graph $G$ and output labeling $f$, decide if $f$ is a valid solution of $\Pi$ in $G$.
    \item Given a graph $G$, find a feasible solution $f$ of $\Pi$ in $G$.
\end{enumerate}
In particular, there is no LCL where merely solving the task is uncomputable; if you can solve a task with full information about $G$ in \ulocal, you can do the same in \clocal. Intuitively, LCLs should not care about computability issues.

Now what we would like to do is to formalize this intuition and prove that we can safely ignore all issues regarding computability in the context of LCL problems; we would like to show that whatever is the complexity of an LCL in \ulocal, it is (at least asymptotically) the same in \clocal. This would mean that the numerous results on LCLs from prior work are compatible with each other regardless of what the specific paper assumes about computability.
Unfortunately, as we will show in this work, this is not the case!

\subsection{Contribution 1: Separation}

Our main contribution is this: we show that there exists an LCL problem $\Pi$ such that:
\begin{itemize}[noitemsep]
    \item in \clocal, any algorithm that solves $\Pi$ requires $\Omega(\sqrt{n})$ rounds,
    \item in \ulocal, there is an algorithm that solves $\Pi$ in $O(\log n)$ rounds.
\end{itemize}
Our key take-home message is that assumptions on computability are of crucial importance already in the study of LCL problems; such issues cannot be safely ignored, results on the complexity landscape cannot be automatically ported between computable and uncomputable worlds, and this aspect needs to be made explicit. (An equally valid interpretation of our work is that LCL problems are a poor proxy for ``reasonable'' graph problems, as we run into issues of computability, which we should never encounter in reasonable graph problems. But currently no good alternatives exist.)

\subsection{Contribution 2: Equivalences}

To fully explain our contributions, it is necessary to draw attention to another subtle issue of the \local model that tends to get swept under the carpet: what do we know about $n$, the size of the input. We will consider three cases that suffice to make the point:
\begin{enumerate}[noitemsep]
    \item no auxiliary information,
    \item all nodes receive as input the same value $N$, and we promise that $n \le N$,
    \item all nodes receive as input the same value $N$, and we promise that $n \le N \le f(n)$ for some fixed, globally-known computable function $f$.
\end{enumerate}
Combining this with models \clocal and \ulocal, we arrive at six combinations that might be all distinct. Yet we show that the landscape shown in \cref{tab:landscape} emerges; six models collapse to three classes $\class{I} < \class{II} \le \class{III}$.

\begin{table}
    \centering
    \caption{For LCL problems, our six models fall in three classes with $\class{I} < {\color{myorange}\class{II}} \le {\color{myblue}\class{III}}$.}\label{tab:landscape}
    \begin{tabular}{lcc}
        \toprule
        & computable & uncomputable \\
        & \local & \local \\
        \midrule
        no auxiliary information
        & $\class{I}$ & \color{myorange} $\class{II}$ \\\addlinespace
        bound $n \le N$
        & \color{myorange} $\class{II}$ & \color{myorange} $\class{II}$ \\\addlinespace
        bound $n \le N \le f(n)$
        & \color{myblue} $\class{III}$ & \color{myblue} $\class{III}$ \\
        \bottomrule
    \end{tabular}
\end{table}

Our LCL problem $\Pi$ serves to separate classes $\class{I}$ and $\class{II}$. However, what seems counterintuitive is that all of these three models in class $\class{II}$ are equally strong for LCLs (and beyond), at least when we restrict ourselves to computable time complexities:
\begin{itemize}[noitemsep]
    \item \ulocal, with no auxiliary information,
    \item \ulocal, with a bound $n \le N$,
    \item \clocal, with a bound $n \le N$.
\end{itemize}
We also show that these two models in $\class{III}$ are equally strong for LCLs (and beyond), for any computable bound $f$, again when we restrict ourselves to computable time complexities:
\begin{itemize}[noitemsep]
    \item \ulocal, with a bound $n \le N \le f(n)$,
    \item \clocal, with a bound $n \le N \le f(n)$.
\end{itemize}
This brings some good news: it turns out that numerous papers that have studied LCLs in the \local model have assumed (explicitly or implicitly) that there is some information on $n$ available, either exact knowledge of $n$ or e.g.\ some linear or polynomial upper bound on $n$. What we show is that as long as some information on $n$ is available, uncomputable and computable worlds collapse into a single model, and we can freely translate results between them.

Finally, we draw attention to the following corollary of our work---these two models are distinct for LCLs:
\begin{itemize}[noitemsep]
    \item \clocal, with no auxiliary information,
    \item \clocal, with a bound $n \le N$.
\end{itemize}
In particular, our problem $\Pi$ can be solved in $O(\log n)$ rounds in \clocal if we know any bound $n \le N$, but it requires $\Omega(\sqrt{n})$ rounds in \clocal if we do not have such information. This result seems counterintuitive: our upper bound $N$ might be arbitrarily high, say, the Ackermann function of $n$, or even the busy beaver \cite{rado-1962-on-non-computable-functions} of $n$, which seems unhelpful, yet it can be so valuable that it suffices to lower the round complexity exponentially.

\section{Technical Overview}

In this section, we give an overview of the techniques and key ideas that we use to prove our main results, and we also provide a road-map to the rest of this paper. Our separation result is presented in Section~\ref{sec:separation}, and our equivalence results are presented in Section~\ref{section:equivalence_results}.

\subsection{\boldmath Separation Between \Clocal and \Ulocal}

First we discuss the proof for the separation between computable and
uncomputable \local where we are not given an upper bound $N$ on the number of
nodes.
The proof is by a careful construction of an LCL problem $\Pi$ that requires
$\Omega(\sqrt{n})$ locality in \clocal but only $\OO(\log n)$ in its
uncomputable counterpart.

Our separation is obtained by constructing a promise-free LCL $\Pi$, where each node outputs several components (or layers). The early layers can be seen as \emph{structure filters} that either certify that the input graph has the intended backbone structure, or they produce a \emph{locally checkable error certificate} that points to a structural violation. Only the nodes that successfully pass all the structural filters are required to solve the final ``hard'' layer.

Formally, each structural layer $\Pi_i$ is split into $\Pi_i^{\mathrm{good}}$ and $\Pi_i^{\mathrm{bad}}$ with disjoint output alphabets. Intuitively
\begin{itemize}
	\item $\Pi_i^{\mathrm{good}}$ The structure ``looks'' correct and nodes output a fixed label.
	\item $\Pi_i^{\mathrm{bad}}$.  The structure is ``broken'' and nodes output an error pointer that locally routes to a witness or failure.
\end{itemize}
We then \emph{compose} these layers so that the subsequent layers' constraints are enforced only inside the region certified as good so far. 

\paragraph{The backbone decomposition.} The overall problem is represented as a composition of five layers
\[
  \Pi = \Pi_\tree \otimes \Pi_\row \otimes \Pi_\grid \otimes \Pi_\turing \otimes \Pi_\consensus
\]
Each layer has a good and bad behavior:
\begin{itemize}
	\item $\Pi_\tree$ certifies a tree-like auxiliary structure by marking nodes as $T$, or else outputs a local pointer to a structural error.
	\item $\Pi_\row$ certifies nodes marked with $G$ form rows that are connected to trees.
	\item $\Pi_\grid$ certifies that the rows together form a grid-like structure.
	\item $\Pi_\turing$ ensures that the labels on the grid encode an execution of a Turing machine.
	\item $\Pi_\consensus$ is the hard layer which gives us the separation. In effect, the rows of the grid need to predict the outcome of the encoded Turing machine.
\end{itemize}

\paragraph{The growing grid structure.} The concrete structure that we would like to ensure is what we refer to as a
\emph{growing grid}. This is an oriented grid where rows are neatly aligned
(see \cref{fig:growing-grids}) and with the additional property that
consecutive rows differ in length by \emph{exactly} one.
Since (virtually) every row has $\Theta(\sqrt{n})$ nodes, to enable
sublinear-time certification inside a row, each row is also equipped with an
attached tree-like structure with $\OO(\log n)$ diameter (see
\cref{fig:row-certificate}).
The tree structure allows each row to operate as a joint computational unit that
is coordinated by the node at the root of the corresponding tree.

\paragraph{The hard problem $\Pi_{\turing}$.} Inside the certified grid region, input labels encode a computation of a universal Turing machine $U$: each row encodes a configuration (tape symbols, head position and state), and local constraints check that consecutive rows represent valid transitions of $U$. If the encoding is inconsistent, nodes can locally point to an error, partitioning the grid into \emph{valid} regions separated by error regions. The crucial final constraint is a \emph{consistency} requirement: whenever a region contains a halting configuration, nodes in that region must output a globally consistent bit indicating whether the halting state is accepting or rejecting.
A subtle point is that we \emph{must} use a machine $U$ that has halting accept
and reject states; see the discussion further below for details.

\paragraph{\Ulocal solves $\Pi$  in $\OO(\log n)$ rounds.} In \ulocal we may
assume access to an oracle $H$ that, given a configuration of the universal
Turing machine $U$, decides whether it halts or not as well as whether the
halting state is accepting or rejecting. 
On a valid grid instance, the root of the tree attached to a row collects the
entire row's encoding in $\OO(\log n)$ rounds, determines the configuration
encoded in the row, and queries $H$. 
Given the oracle's answer, the root node sets the output for the nodes it
controls to the final state of $U$ or, if the configuration is not halting, to
$0$.
Since all the rows access the same oracle and the universal Turing machine is
deterministic, the output of all rows is consistent across the same component of
the input graph.

\paragraph{\Clocal needs $\Omega(\sqrt{n})$ rounds.}
The lower bound relies on a diagonalization argument.
Assume that there is a \clocal algorithm $\AA$ that purportedly solves $\Pi$ in
$o(\sqrt{n})$ rounds.
On our growing-grid instances $G_k$ with $k$ rows (and thus $\Theta(k^2)$
nodes), such a locality is $o(k)$; hence a node in a fixed early layer cannot
see the bottom boundary of the grid.
In fact, we can show that, for large enough values of $k$, nodes in these early
layers must commit to an output in a \emph{constant number of rounds}.
We exploit this by constructing a Turing machine $T_\AA$ that, given an input
$x$, simulates $\AA$ on $G_k$ for a sufficiently large $k$ and produces the
output bit of $\AA$ on the row with configuration corresponding to $x$.
Notice that $\AA$ always halts, hence so does $T_\AA$.
Finally, we feed $T_\AA$ to a diagonalizer $D$ that, on input $\descr{T_\AA}$,
outputs the opposite bit whenever $T_\AA$ halts.
Thus, given the input $x = (\descr{D}, \descr{T_\AA})$, the nodes of $\AA$ in
the row corresponding to $x$ will produce the opposite of what is required.
It follows that $\AA$ fails to solve $\Pi$, implying its complexity in \clocal
is $\Omega(\sqrt{n})$.

\paragraph{A remark about the specification of the problem output.}
Let us take a closer look at the condition for the output bit we ask for in our
problem $\Pi_{\turing}$:
\begin{enumerate}
  \item If $U$ halts in the final configuration of the growing grid, then output
  $0$ if it rejects or $1$ if it accepts.
  \item Otherwise (i.e., if $U$ has not halted by the final configuration), then
  output either $0$ or $1$ consistently.
\end{enumerate}
Interestingly, this condition must be very carefully chosen.
For instance, suppose we leave the problem description as is and only replace
the second point with the following:
\begin{enumerate}
  \item[2a.] Otherwise (i.e., if $U$ has not halted by the final configuration),
  then output $0$.
\end{enumerate}
As it turns out, \emph{this breaks the $\OO(\log n)$ upper bound}:
Of course, with access to an uncomputable oracle, we can still correctly
determine whether $U$ halts or not; however, the issue is that we must also
determine \emph{if the input grid is large enough} to contain the entire trace
of $U$ until it halts, which inevitably requires inspecting the entire grid.

A much more counterintuitive situation applies to the first of the two
requirements for the output bit.
It is a standard result in computability theory that, without restriction, we
may define Turing machines such that they are only allowed to halt if they are
in an accepting state (and enter an infinite loop whenever an input should be
rejected).
Being aware of this, we might be tempted to replace our $U$ that has explicit
accept and reject states with another universal $U'$ that only halts in an
accepting state.
This simplifies the first of the two items above to the following:
\begin{enumerate}
  \item[1a.] If $U'$ halts in the final configuration of the growing grid, then
  output $1$ everywhere.
\end{enumerate}
Surprisingly, this now \emph{breaks the lower bound}:
Together with the original second item, the resulting problem is trivial since
one can always output $1$ no matter what $U'$ does.
Replacing 2 with 2a does not remedy this as we then run into the same issue with
the upper bound as before (i.e., we fail to obtain an $\OO(\log n)$ algorithm in
\ulocal).
Hence it is \emph{crucial that our $U$ has both halting accept and reject
states}, unlike what our intuition from computability theory might suggest.

\paragraph{\boldmath An alternative upper bound in \clocal with an
upper bound on $n$.}
As an addendum, we also describe a strategy for solving $\Pi$ in
\clocal with $\OO(\log n)$ when an upper bound $N$ (no matter how bad)
on $n \le N$ is given.
It is true that one can use our equivalence results to obtain such an algorithm,
but it is instructive to come up with a direct strategy to grasp why the
knowledge of $N$ makes such a big difference.

The idea is to proceed as in the \ulocal algorithm but replace the
oracle $H$ with a computable procedure.
More specifically, using $N$ we infer a maximal number of steps $t_N$ that $U$
can run for; otherwise its trace would not fit into the input instance we are
given.
Then we can simulate $U$ until $t_N$ steps have elapsed and, if $U$ has not
halted yet, safely output $0$ on all nodes.
The point is that, given $N$, \emph{we have a bound on how large our instance
can be} (even if it is arbitrarily bad!), which allows us to circumvent the
computability issues; \emph{without $N$, we cannot know how large our instance
actually is}, and thus it might always be large enough to contain the entire
trace of $U$ until it halts---which we cannot tell for sure without inspecting
the whole instance.

\subsection{Equivalence Results}

Next we discuss our equivalence results.
Recall there are two equivalences that we prove:
\begin{enumerate}
  \item between \clocal with an upper bound $N \ge n$ (where possibly
  also $N \le f(n)$ for some computable $f$) and \ulocal; and
  \item \ulocal with and without such an upper bound $N$.
\end{enumerate}
Before we turn to these equivalences, we first discuss the concept of a
\emph{maximally safe neighborhood}, which is used in both proofs.

\paragraph{Maximally safe neighborhoods.} In the ball-growing view of \local model, a deterministic algorithm can be described by a function that given the currently gathered rooted neighborhood $N^t(v)$ and the current radius $t$, either outputs a label or returns a sentinel $\bot$ to indicate ``keep exploring''; the node halts at the first radius $t$ for which the function returns a label. A time bound $T(n)$ means that on every $n$-node network, every node must halt by round $T(n)$.  A subtle detail is that nodes do not know $n$. Hence, whether a node is allowed to still be running at round $t$ must be justified purely from its local view. This motivates a notion of when a radius-$t$ view is consistent with a global bound $T$.
\begin{definition}[Informal]\label{def:tmp_def_1}
	For a network $G$ and a node $v \in V(G)$, we say $t \in \N_0$ is
	\emph{$T$-safe} for $v$ if there is no network $G'$ and $v'
	\in V(G')$ such that the $t$-neighborhoods rooted at $v$ and $v'$ are
	isomorphic and $t > T(\abs{V(G')})$.
	In addition, $t$ is said to be \emph{maximally $T$-safe} for $v$ if it is maximal
	with this property (i.e., $t$ is safe but $t+1$ is not).
\end{definition}
The intuition behind Definition~\ref{def:tmp_def_1} is that if $N^t(v)$ is $T$-safe, then ``still not having halted before round $t$'' is not in conflict with the promise that all $n'$-node instances must halt by round $T(n')$. Indeed, the same local situation $N^t(v)$ cannot occur in any instance whose time bound $T(n')$ is already strictly less than $t$. Conversely, if a new view at radius $t$ could also occur inside some smaller instance with $T(n')<t$, then any rule that keeps running when seeing that view would necessarily violate the time bound on that smaller instance. That is because the node cannot distinguish the two situations.

Maximally $T$-safe neighborhoods are used as a means to define a \emph{stopping view} that depends only on $T$ and is independent of the algorithm. Moreover, when $T$ is non-decreasing and computable, whether a given rooted neighborhood is maximally $T$-safe is decidable (see Lemma~\ref{lem:max-safe-re}), which is what lets us algorithmically enumerate and recognize these views when needed.
\paragraph{From uncomputable to computable with an upper bound on $n$.} The proof of this equivalence uses the following observation:
\begin{quote}
	When nodes know an upper bound $N$ on the number of vertices, there are only finitely many maximally $T$-safe neighborhoods that can arise in networks of size at most $N$.
\end{quote}
This means that if there is an LCL problem $\Pi$ and an  \ulocal  algorithm $\AA$ that has an upper bound $N$ on the network size $n$, then we  can define a \clocal algorithm $\BB$ that solves $\Pi$ with the same locality and knowledge of the same upper bound $N$ on the number of nodes.  The computable algorithm $\BB$ simply enumerates all these finitely many maximally $T$-safe neighborhoods and searches for a mapping from them to output labels that satisfies the LCL constraints (see Theorem~\ref{thm:equiv-comp-uncomp}).

\paragraph{\boldmath Removing the upper bounds on $n$.} Given a family $\{\AA_{N}\}$ of algorithms that work with knowledge of $N$, we ``sandbox'' each $\AA_{N}$ by first growing the view until it becomes maximally $T$-safe, then simulating $\AA_N$ inside that neighborhood, and outputting $\bot$ if the simulation ever tries to inspect nodes outside it. The key point is that a maximally safe neighborhood cannot be enlarged without risking a violation of the time bound, so this normalization stays within $T(n)$ rounds and makes the later stabilization argument possible (see Theorem~\ref{thm:eliminate-upper-bound}).

\paragraph{A final remark.} Our equivalence results combine two different ingredients. The conversion from an \ulocal algorithm to a \clocal one given an upper bound is constructive. Indeed, we can enumerate the (finitely many) maximally $T$-safe neighborhoods and search for a correct mapping to output labels.  In contrast, removing the need for an upper bound in \ulocal relies on an existential stabilization step: we pass to an infinite subsequence of bounds on which the sandboxed algorithms agree on each neighborhood  (see the Construction step $2$ in \cref{thm:eliminate-upper-bound}). We do not currently know a direct constructive transformation that replaces this stabilization step.

\section{Preliminaries}

We write $\N_+$ for the set of positive integers and $\N_0$ for $\N_+ \cup
\{0\}$. Moreover, we use the notation $[n] =\{1,\dots, n\}$ for any $n\in \N_+$.

\paragraph{Graphs.}
A graph $G=(V,E)$ consists of a set of nodes $V$ and a set of edges $E$, and we
use the notation $V(G)$ and $E(G)$, respectively, if we need to specify which
graph we refer to.
For any two nodes $u,v\in V$ in a graph $G=(V,E)$, $\dist_G(u,v)$ is the length
of a shortest path starting from $u$ and ending at $v$.
We drop the subscript $G$ whenever there is no ambiguity.
The degree of a node $v \in V$ is denoted by $\deg_G(v)$. If $G$ is a subgraph of $H$, we write $G\subseteq H$. For any subset of nodes $A\subseteq V$, we denote by $G[A]$ the subgraph induced by the nodes in $A$. For any nodes $u,v\in V$, $\dist_G(u,v)$ is the length
of a shortest path starting from $u$ and ending at $v$; if $u$ and $v$ are disconnected, then $\dist_G(u,v) = +\infty$.
The degree of a node $v \in V$ is denoted by $\deg_G(v)$. We drop the subscript $G$ whenever there is no ambiguity.

Since we are considering distributed models, our definition of the radius-$r$
neighborhood of a node is tailored to match precisely what can be gathered in
$r$ rounds of a distributed message passing protocol (with unbounded
communication).
It is indeed very similar to the standard graph-theoretic one, but there is a
fine distinction to be made, which we discuss after the definition.

\begin{definition}[Neighborhoods]
  Let $G$ be a graph and $v \in V(G)$ a node of $G$.
  For $r \in \N_0$, the \emph{radius-$r$ neighborhood} $\NN^r(v)$ of $v$ is a
  labeled graph constructed as follows:
  \begin{itemize}[noitemsep]
    \item $\NN^r(v)$ contains every node at distance at most $r$ from $v$; that
    is, the set of nodes in $\NN^r(v)$ is exactly $\{ u \in V(G) \mid \dist(u,v)
    \le r \}$.
    \item $\NN^r(v)$ contains every edge $u_1u_2 \in E(G)$ where either
    $\dist(u_1,v)$ or $\dist(u_2,v)$ is at most $r-1$.
    \item Every node in $\NN^r(v)$ is labeled with its degree $\deg(u)$ in $G$.
  \end{itemize}
\end{definition}

Naturally, the graph-theoretic definition would state that $\NN^r(v)$ is the
subgraph of $G$ induced by nodes at distance at most $r$ from $v$.
However, if we are in a distributed message passing network, then a node does
know its number of neighbors but not necessarily \emph{with which node} it is
connected to until it has actually learned unique information about it (e.g.,
the neighbor's unique identifier).
Extended to the radius-$r$ case, this means the node $v$ in the definition above
does not know about edges between nodes that are at distance (exactly) $r$ from
it; it knows that these edges exist (since it knows the degrees of these nodes)
but not how the edges are connected.

We consider the class of all graphs of maximum degree $\Delta$, formally
\[
\GG = \left\{G=(V,E)\mid \max_{u\in V}\deg(u)\leq \Delta\right\}.
\]
Such a graph class is \emph{recursively enumerable}. Throughout this paper we consider $\Delta = 10$, but we stress that choice of the graph family does not matter either that much. We make use of this graph class to show our separation results in \cref{sec:separation}. But for the equivalence results (\cref{section:equivalence_results}), one can use any bounded-degree recursively enumerable graph family.

\paragraph{Networks.}

Fix a finite set of \emph{input labels} $\Sigma_\iin$.
A \emph{network} is a tuple $(G,x,\ddeg,\id,\{p_v\})$ where:
\begin{itemize}[noitemsep]
  \item $G$ is a graph,
  \item $x \colon V(G) \to \Sigma_\iin$ is an \emph{input} to the network,
  \item $\ddeg\colon V(G) \to \N_0$ is a \emph{degree function} such that, for
  every node $v \in V(G)$, $\ddeg(v) = \deg(v)$; and
  \item $\id\colon V(G) \to [n^c]$ is an injective assignment of unique
  \emph{identifiers} for nodes $V(G)$.
  \item $\{p_v\}$ is a family of \emph{port-numbering functions} such that for every node $v$, function $p_v$ is a bijection from edges adjacent to $v$ to port numbers $[\ddeg(v)]$.
\end{itemize}
Here $c$ is a positive constant chosen arbitrarily; for the purposes of this
paper, we assume that $c$ is a \emph{fixed} constant that is known by all
models.
We can encode inputs on edges by allowing the input label set $\Sigma_\iin$ to refer to port numbers.

We generalize neighborhoods to networks.
Given a neighborhood $\NN^r(v)$ of a vertex $v$, the \emph{$r$-neighborhood rooted
at $v$} is the network $(\NN^r(v),x\restriction_{\NN^r(v)},\ddeg\restriction_{\NN^r(v)},\id\restriction_{\NN^r(v)}, \{p_v\}_{v \in \NN^r(v)})$
obtained by restricting the network $(G,x,\ddeg,\id,\{p_v\})$ to $\NN^r(v)$.
In this context, we imagine the node $v$ is distinguished in some way (e.g.,
by extending $x$ with a special label or by extending the triple by an element
designating $v$) and refer to it as the \emph{root} of the neighborhood.
In case it is not important to indicate which node specifically is the root, we
also write simply \emph{rooted neighborhood} to refer to the concept in general.
When $\Sigma_\iin$ is fixed, we write $\rootn$ for the set of all rooted
neighborhoods.
Abusing notation, we use $\NN^r(v)$ to refer both to a \emph{graph}, that is, the
neighborhood of $v$ (in the graph $G$) and to a \emph{network}, that is, the
$r$-neighborhood rooted at $v$.

The concept of isomorphism also readily extends to networks.
Two networks $(G,x,\ddeg,\id,\{p_v\})$ and $(G',x',\ddeg',\id',\{p_v'\})$ are \emph{isomorphic} if
there is a bijective map $\phi\colon V(G) \to V(G')$ that is an isomorphism
of graphs and additionally satisfies $x \comp \phi^{-1} = x'$, $\deg \comp \phi^{-1}
= \deg'$, $\id \comp \phi^{-1} = \id'$, and $p_v(\{v,u\}) = p'_{\phi(v)}\bigl(\{\phi(v),\phi(u)\}\bigr)$ for every $u,v \in V(G)$.
Using this same notion, we occasionally also say that two subgraphs $H \subseteq
G$ and $H' \subseteq G'$ are \emph{isomorphic} (\emph{in the extended sense}) if
we have such an isomorphism $\phi$ between the networks
$(H,x\restriction_H,\ddeg\restriction_H,\id\restriction_H,\{p_v\}_{v\in V(H)})$ and $(H',x'\restriction_{H'},\ddeg'\restriction_{H'},\id'\restriction_{H'},\{p'_v\}_{v\in V(H')})$ (with no
explicit reference to these networks).

\paragraph{Labeled graphs and locally checkable labelings.}
Following the notation in~\cite{balliu-ghaffari-etal-2025-shared-randomness-helps-with}, we define the notion of labeled graph and locally checkable labeling (LCL) problem.

\begin{definition}[Labeled graph]
	Let $\mathcal{V}$ and $\mathcal{E}$ be sets of labels and a graph $G=(V,E)$ is $(\mathcal{V},\mathcal{E})$ labeled if every vertex $v\in V$ is assigned a label from $\mathcal{V}$ and each pair $(u,e)\in V\times E$ such that $u\in e$, is assigned a label from $\mathcal{E}$.
\end{definition}
A node-edge pair $(u,e)$ that satisfies $u\in e$ is also called a \emph{half-edge} incident to $u$.
\begin{definition}
	Let $G$ be a graph, and let $\CC$ be a set of constraints over the labels $\VV$ and $\EE$. The graph $G$ satisfies $\CC$ if  and only if
	\begin{itemize}[noitemsep]
		\item $G$ is $(\VV,\EE)-$labeled, and
		\item the constraints of $\CC$ are satisfied over all nodes of $G$.
	\end{itemize}
\end{definition}
\begin{definition}[LCL, Definition~2.3 in~\cite{balliu-ghaffari-etal-2025-shared-randomness-helps-with}]
	A locally checkable labeling problem $\Pi$ is defined by a tuple $(\mathcal{V}_{\texttt{input}},\mathcal{E}_{\texttt{input}},\mathcal{V}_{\texttt{output}},\mathcal{E}_{\texttt{output}},\mathcal{C} )$ where $\mathcal{C}$ is the set of constraints and given a  $(\mathcal{V}_{\texttt{input}},\mathcal{E}_{\texttt{input}})$-labeled  input graph $G$, the aim is to $(\mathcal{V}_{\texttt{output}},\mathcal{E}_{\texttt{output}})$-label $G$ such that the labelings satisfy the constraints $\mathcal{C}$.
\end{definition}

Given a graph $G$ which is $(\mathcal{V},\mathcal{E})$-labeled, let $P = (u=v_1,v_2,\ldots,v_{k+1})$ be a path in $G$. For each edge $e = (v_i,v_{i+1})$ in $P$, the half-edge $(v_i,e)$ is labeled with $L_{v_i}(e) = L_i$. Then

\[
f(u,L_1,L_2,\ldots, L_k) = \begin{cases}
	v_{k+1}, & \text{if $P$ exists and is unique}\\
	\bot, & \text{otherwise.}
\end{cases}
\]
Intuitively, function $f$ allows us to start from a node and follow a path defined by unique edge labels to find the other end of the path; we will use this to define LCL problems.

\paragraph{Models of computation.}
The \local model~\cite{linial-1987-distributive-graph-algorithms-global} is a distributed model of computation in which each node in the network executes the same algorithm independently. Computation proceeds in synchronous rounds, and in each round every node may exchange messages with all of its neighbors. There is no bound on the size of the messages (unbounded communication). Every node must eventually terminate and produce an output. The complexity measure in this model is the number of rounds required until all nodes output.

We define the \local model in the so-called \emph{ball growing} formalism.
The intuition for this model is the following:
Suppose we are considering the execution of our distributed algorithm from the
perspective of a node $v \in V(G)$.
The node has witnessed $t \in \N_0$ rounds of the algorithm but has not
committed to an output yet.
Since in \local we have access to unbounded communication, without restriction
during these $t$ rounds the node $v$ has gathered complete information about the
rooted neighborhood $\NN^t(v)$ corresponding to it.
At this point, it needs to make a decision:
Does it commit to an output based on what it has seen so far (i.e., $\NN^t(v)$),
or does it explore the graph a bit further, \emph{growing} $\NN^t(v)$ to
$\NN^{t+1}(v)$?

\begin{definition}[Ball-growing \local]
  Fix label sets $\Sigma_\iin$ and $\Sigma_\out$.
  The \local model in the \emph{ball-growing} formalism is defined by a map
  $f\colon \rootn \times \N_0 \to \Sigma_\out \cup \{ \bot \}$, where $\bot
  \notin \Sigma_\out$ is a sentinel value.
  Given a network $(G,x,\ddeg,\id,\{p_v\})$ as input, the output of a node $v \in V(G)$
  is the value $f(\NN^t(v),t)$ where $t \in \N_0$ is minimal such that
  $f(\NN^t(v),t) \neq \bot$.
  This can be thought of as the result of executing the following iterative
  process at each node $v$:
  \begin{enumerate}[noitemsep]
    \item Initialize $t$ to zero.
    \item If $f(\NN^t(v),t) \neq \bot$, output it; otherwise increment $t$ and
    repeat.
  \end{enumerate}
  Without restriction, we consider only \enquote{sensible} algorithms where, if
  $\NN^t(v)$ has radius $t$, then $f(\NN^t(v), t') = f(\NN^t(v),t+1) \neq \bot$ for
  every $t' \ge t+1$.

  Without further qualifications, the above gives us the \ulocal model.
  If we also insist that $f$ is a computable function, then the resulting model
  is \clocal.
  The \emph{complexity} or \emph{locality} of a \local algorithm (in either
  case) is the minimal function $T\colon \N_+ \to \N_0$ such that, if $G$ is a
  graph on $n$ nodes and $v \in V(G)$ is a node in a network $(G,x,\ddeg,\id,\{p_v\})$,
  then there is $t \le T(n)$ for which $f(\NN^t(v),t) \neq \bot$.
\end{definition}
Note that equivalently, $T(n)$ can be seen as the worst-case number of rounds required over all graphs of size $n$. It follows that $T$ is non-decreasing: if $n_1 \leq n_2$, then the set of graphs with $n_1$ nodes is contained in the class of graphs with $n_2$ nodes, and taking the maximum over a larger class cannot yield a smaller value. Hence, $T(n_1) \;\leq\; T(n_2)$.

\section{Separation Result}\label{sec:separation}
In this section, we prove the following result:

\begin{theorem}\label{thm:separation}
  There is a problem $\Pi$ that has round complexity $\Omega(\sqrt{n})$ in
  \clocal but only $\OO(\log n)$ in \ulocal.
\end{theorem}

\subsection{LCL Backbone}
\label{sec:lcl-backbone}

Before we present the definition of our LCL problem, we introduce a few concepts
that are helpful in keeping the description more digestible.
This approach is implicit in many works~\cite{balliu-ghaffari-etal-2025-shared-randomness-helps-with}; we make
the concept explicit here in order to better structure the presentation.

The first of these concepts is the notion of a \emph{certification LCL problem}.
This is an LCL problem $\Pi = (\Vin, \Ein, \Vout, \Eout, \CC)$
that actually consists of two sub-problems $\Pi^\good$ and $\Pi^\bad$ where the
label sets of $\Pi^\good$ and $\Pi^\bad$ are incompatible with each other.
More precisely, we have the following:
\begin{itemize}
  \item $\Pi^\good$ and $\Pi^\bad$ use the same sets $\Vin$ and $\Ein$ as input
  labels.
  \item As output labels, $\Pi^\good$ uses $\Vout_\good$ and $\Eout_\good$ while
  $\Pi^\bad$ uses $\Vout_\bad$ and $\Eout_\bad$, where we have $\Vout_\good \cap
  \Vout_\bad = \varnothing$ and $\Eout_\good \cap \Eout_\bad = \varnothing$.
  The output label sets of $\Pi$, that is, $\Vout$ and $\Eout$ are the resulting
  disjoint unions thereof, respectively.
  \item $\Pi^\good$ and $\Pi^\bad$ specify different constraint sets $\CC_\good$
  and $\CC_\bad$.
  The constraint set $\CC$ of $\Pi$ is the union $\CC_\good \cup \CC_\bad$.
\end{itemize}
Whenever we deal with such an LCL problem $\Pi$, we write $\Pi = (\Pi^\good,
\Pi^\bad)$ to refer to these two parts.
(In this definition as presented we make no distinction between $\Pi^\good$ and
$\Pi^\bad$ yet; this only appears in connection with the next concept.)

The second concept that we need is the notion of \emph{composing} LCL problems
where one of the problems is a certification LCL problem.
Given a certification LCL problem $\Pi_1 = (\Vin_1, \Ein_1, \Vout_1, \Eout_1,
\CC_1)$ and another (general) LCL problem $\Pi_2 = (\Vin_2, \Ein_2, \Vout_2,
\Eout_2, \CC_2)$, we define $\Pi_1 \otimes \Pi_2$ to be the LCL problem $\Pi$
with the following characteristics:
\begin{itemize}
  \item The sets of input labels of $\Pi$ are $\Vin_1 \times \Vin_2$ and $\Ein_1
  \times \Ein_2$ for nodes and edges, respectively.
  \item The sets of output labels are $\Vout_1 \times (\Vout_2 \cup \{ \bot \})$
  and $\Eout_1 \times (\Eout_2 \cup \{ \bot \})$, where $\bot$ is a sentinel
  value that is not an element of $\Vout_2$ or $\Eout_2$.
  \item The constraints of $\Pi$ are as follows:
  \begin{itemize}
    \item An output $(x,y)$ of a node or half-edge is only admissible if it is
    one of the following:
    \begin{enumerate}
      \item If $x$ is an output label from $\Pi_1^\good$, then $y$ is an output
      label from $\Pi_2$ (i.e., $y \neq \bot$).
      \item If $x$ is an output label from $\Pi_1^\bad$, then $y = \bot$.
    \end{enumerate}
    \item If we project the outputs to their first component everywhere, then
    the resulting graph satisfies $\CC_1$.
    \item If we project the nodes' output to their second component everywhere
    and disregard nodes and edges that have $\bot$ as output (where we disregard
    the entire edge if either half-edge is labeled $\bot$), then the resulting
    graph satisfies $\CC_2$.
  \end{itemize}
\end{itemize}
This definition naturally extends to the composition of multiple problems $\Pi_1
\otimes \dots \otimes \Pi_k$ (where $\Pi_i$ is a certification LCL problem for
every $i < k$) in the obvious manner.

The idea behind these two formalisms is to use the certification LCL problems as
a \enquote{filter} that only allows instances with a structure compatible with
the certification scheme of $\Pi_i^\good$ to pass through.
Undesirable structures can be detected and the offending parts certified using
labels from $\Pi_i^\bad$.
(The latter is required in order for the problem to not be trivially solvable.)
This approach enables us to turn a problem $\Pi_k$ that gives us a separation
assuming a certain structure into a promise-free separation by embedding the
certification of the structure into the problems $\Pi_1,\dots,\Pi_{k-1}$.

We now give an overview of the LCL problem that we construct.
Namely our problem is
\[
    \Pi = \Pi_\tree \otimes \Pi_\row \otimes \Pi_\grid \otimes \Pi_\turing \otimes \Pi_\consensus
\]
where the individual problems are as follows:

\begin{enumerate}
  \item In $\Pi_\tree$, we certify a complete bipartite tree structure
  everywhere the nodes receive as input label a special symbol $\mathrm{T}$.
  The certificate comes from the sub-problem $\Pi_\tree^\good$.
  If any connected component formed by nodes with this input does not form a
  complete bipartite tree, then we can output a certificate for this (using
  labels from $\Pi_\tree^\bad$) that can be checked with constant radius.
  Nodes that are not marked with $\mathrm{T}$ are quiescent; any output is valid for them.

  \item $\Pi_\row$, is similar to $\Pi_\tree$ in that we are also certifying a structure in nodes marked with $\mathrm{G}$, in this case that an oriented row is properly attached to the leaves of a tree.
  In particular, nodes satisfying $\Pi_\row^\good$ form rows that have beginnings and ends, and these match the left-most and right-most nodes of the attached trees.
  Also, every row has only one tree attached to it, which means that all nodes of the row are within distance $O(\log n)$ of each other.
  If this does not hold, then the algorithm can solve $\Pi_\row^\bad$ which certifies that there exists an error in the row.

  \item With $\Pi_\grid$, we continue our certification of the structure of nodes marked with $\mathrm{G}$, this time certifying that the nodes form an oriented grid where each row is larger than the previous one by \emph{exactly one node} at the far right end of the row, which we name the \emph{growing grid}.  
  Again, if the structure is correct, then we can certify it using labels from
  $\Pi_\grid^\good$; if the opposite is the case, then we can show this using
  labels from $\Pi_\grid^\bad$.

  \item The problem $\Pi_\turing$ embeds the trace of a universal Turing
  machine $U$ in the growing grid structure evolving in the downward direction.
  It contains specifications such that every configuration of $U$ corresponds to
  a row in the grid.
  This includes ensuring that the configuration has a single Turing machine head
  (which is checked using the tree structure) and that the configuration of the
  respective row arises after one step of $U$ from the one in the previous row
  (which can be checked locally).
  As before, any detected violations can be certified using the sub-problem
  $\Pi^\bad_\turing$.
  Otherwise the algorithm may solve $\Pi^\good_\turing$.
 
  \item The last problem $\Pi_\consensus$ is the one where we get a separation between computable and uncomputable versions of the \local model.
  In every connected component, if the configuration where $U$ halts is present,
  then we must consistently output a bit corresponding to the accept or reject
  state of $U$; otherwise, that is, if $U$ has not halted yet in the trace that
  we are given, then we are allowed to output any of the two bits as long as it
  is consistent across the component.
\end{enumerate}

\begin{figure}
  \centering
  \begin{subfigure}{0.48\textwidth}
    \centering
    \includegraphics[scale=0.5]{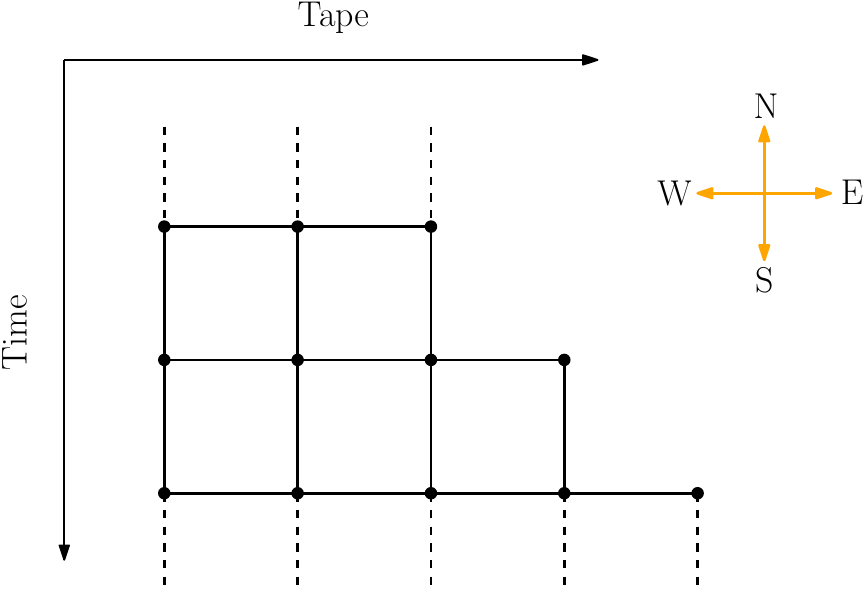}
    \caption{A fragment of a growing grid. The tape grows towards right, and time flows downward.}
    \label{fig:growing-grids}
  \end{subfigure}\hfill
  \begin{subfigure}{0.5\textwidth}
    \centering
    \includegraphics[scale=1.1]{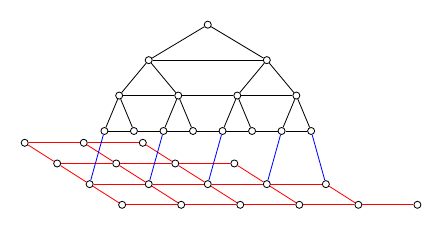}
    \caption{A tree growing on top of a row of a growing grid. In reality, all rows have a tree on top of them.}
    \label{fig:row-certificate}
  \end{subfigure}
  \caption{Visualizations of growing grids (\cref{def:growing-grid}). A tree attached to a row is shown on right.}
  \label{fig:lcl-construction}
\end{figure}
See \cref{fig:lcl-construction} for an illustration of the structures just
described.
In the next section, we formalize this problem and show how to carry out the
certification procedures we just mentioned.

\subsection{Defining the LCL}

\subsubsection{Trees}
We start by lifting certification of tree-like structures from existing work \cite{balliu-ghaffari-etal-2025-shared-randomness-helps-with}.
\begin{definition}[Perfect Tree-like Structure]\label{def:perfect-tree-like structure}
A graph \( G \) is said to be a \emph{perfect tree-like structure} of height \( \ell \) if there exists an assignment of coordinates to its vertices satisfying the following conditions:
\begin{enumerate}
    \item Each node \( u \in V(G) \) is assigned coordinates \((\ell_u, k_u)\), where
    \begin{itemize}
        \item \( 0 \le \ell_u < \ell \) denotes the \emph{level} (or depth) of \(u\) in the tree, and
        \item \( 0 \le k_u < 2^{\ell_u} \) denotes the \emph{horizontal position} of \(u\) within its level, assuming a level-by-level traversal.
    \end{itemize}

    \item Two nodes \( u = (\ell_u, k_u) \) and \( v = (\ell_v, k_v) \) are adjacent in \(G\) if and only if one of the following holds:
    \begin{enumerate}
        \item \(\ell_u = \ell_v\) and \(|k_u - k_v| = 1\) \hfill (horizontal edges),
        \item \(\ell_v = \ell_u + 1\) and \(k_u = \left\lfloor \tfrac{k_v}{2} \right\rfloor\) \hfill (parent–child relation), or
        \item \(\ell_u = \ell_v + 1\) and \(k_v = \left\lfloor \tfrac{k_u}{2} \right\rfloor\) \hfill (child–parent relation).
    \end{enumerate}

    \item All levels are fully populated, i.e. for each $0\leq i < \ell$, there exist nodes with coordinates $(i,j)$ for each $j\in \{0,\ldots,2^i-1 \}$.
\end{enumerate}
\end{definition}

\paragraph{Good perfect tree.}
The LCL problem that certifies perfect tree-like structures is denoted by $\Pi_\tree = (\Pi_\tree^\good,\Pi_\tree^\bad)$. The LCLs $\Pi_\tree^\good$ and $\Pi_\tree^\bad$ are defined as follows:
\[
	\Pi_\tree^\good =
(\Vin_\tree, \Ein_\tree,
 \Vout_\goodtree, \Eout_\goodtree,
 \CC_\goodtree),
\]
where $\Vin_\tree = \{\mathrm{T}, \mathrm{G}\},\ \Ein_\tree = \{\chl, \chr, \Parent, \Left\}$ and $\Vout_\tree = \{\bot\}$ and $\Eout_\tree = \{\bot\}$.
The constraints of $\CC_\goodtree$ are the same set of local constraints of tree-like structure ($\mathcal{C}^{\mathrm{tree}}$) specified in~\cite{balliu-ghaffari-etal-2025-shared-randomness-helps-with}, except that we enforce them only for nodes with input label $\mathrm{T}$.
We reiterate them here:

\begin{framed}
	\noindent\textbf{Constraints $\CC_\goodtree$}

	\noindent
	Every node with input label $\mathrm{T}$ and output label $\bot$ satisfies the following constraints:
	\begin{enumerate}
    		\item
    		For any node $u$ and two distinct incident edges $e,e'$, it holds that $L_u(e) \neq L_u(e')$.

    		\item
    		For any edge $e = \{u,v\}$, if $L_u(e) = \Left$, then $L_v(e) = \Right$, and vice versa.

    		\item
    		For any edge $e = \{u,v\}$, if $L_u(e) = \Parent$, then $L_v(e) \in \{\chl,\chr\}$, and conversely.

    		\item
    		If $u$ has an incident edge $e = \{u,v\}$ such that $L_u(e)=\Parent$ and $L_v(e)=\chl$,
    		then $f(u,\Parent,\chr,\Left)=u$.

    		\item
    		If $u$ has incident edges $\{u,v\}$ and $\{u,w\}$ with
    		$L_u(\{u,v\})=\Parent$, $L_v(\{u,v\})=\chr$, and $L_u(\{u,w\})=\Right$,
    		then $f(u,\Parent,\Right,\chl,\Left)=u$.

		\item If a node has an incident half edge labeled $\chl$ then it must also have an incident half edge labeled $\chr$ and vice versa.

    		\item
   		 A node has no incident half-edge labeled $\Parent$
    		if and only if it has no incident half-edges labeled $\Left$ or $\Right$.

        	\item If a node $v$ has no incident half-edges labeled $\chl$ or $\chr$,
				then neither do $f(v,\Left)$ nor $f(v,\Right)$, if they exist. Also, such a node $v$ can have at most one incident half-edge labeled $\mathrm{Ch}$.
		\item If a node $u$ has an incident edge $e= \{u,v\}$ with label $L_u(e)=\Parent$ and $L_v(e) = \chr$ (respectively $L_v(e)=\chl$), then $u$ has an incident edge labeled $\Right$ (respectively $\Left$) if and only if $f(u,\Parent)$ has an incident edge labeled $\Right$ (respectively $\Left$).
	\end{enumerate}
\end{framed}

In the following theorem, we claim that the constraints not only certify tree-like structures, but certify perfect tree-like structures.
\begin{theorem}
	A graph $G= (V,E)$ is a perfect tree-like structure if and only if it admits a $(\Vin_\tree,\Ein_\tree)$-labeling which satisfies  the constraints $\CC_\goodtree$.
\end{theorem}

\begin{proof}
	Let us assume that $G$ is perfectly tree-like and it has $\ell$ levels. We refer to the vertices of $G$ with the same coordinate-convention as the definition of perfect tree-like. The labeling is done as follows: Let node $u$ have the coordinate $(i,k)$.
	\begin{enumerate}[noitemsep]
		\item For node $v = (i,k+1)$ and $e = \{u,v\}$, set $L_u(e) = \Right$ and $L_v(e) = \Left$.
		\item For node $v = (i+1,2k)$, set $L_{u}(e) = \chl$ and $L_{v}(e) = \Parent$.
		\item For node $v = (i+1,2k+1)$, set $L_u(e) = \chr$ and $L_v(e) = \Parent$.
	\end{enumerate}
	It is not hard to check that the resultant labeling satisfies the constraints $\labels{C}{good}{tree}$.

	For the other direction, let us assume that the graph $G$ admits a $(\Vin_\tree,\Ein_\tree)$-labeling that satisfies $\CC_\goodtree$. We claim that the graph is perfect tree-like. The fact that the graph is tree-like was established in~\cite[Lemma~4.2]{balliu-ghaffari-etal-2025-shared-randomness-helps-with}. We claim that every level in this tree is fully populated. We proceed inductively. Let us assume that the tree has $\ell$ levels. For $\ell = 2$, rules~1 and~2 ensure that the tree is fully populated at each level. Assume this to be true for some value $\ell = k+1$ where $k\geq 1$.
	Let us assume that the layer $i$ is not full.
	Then there must exist a left-most node $(i,k)$ such that the tree is proper to left of it, and towards right the level is incomplete.
	If $(i,k)$ is a left child of its parent, then by rule~6 its parent must also have a right child which is the right sibling of $(i,k)$, contradicting the assumption that the row is broken right of $(i,k)$.
	If node $(i,k)$ is instead a right child of its parent, then its parent must have a right sibling too by rule 9, as otherwise the row would be complete.
	But then by rule 5 node $(i,k)$ also has a right sibling that is the left child of its parent, again contradicting the fact that the row is broken right of $(i,k)$.
\end{proof}

\paragraph{Bad perfect tree.}
The LCL problem $\Pi_\tree^\bad$ is defined by the tuple
\[
	\Pi_\tree^\bad =
(\Vin_\tree, \Ein_\tree,
 \Vout_\badtree, \Eout_\badtree,
 \CC_\badtree),
\]
where
$\Vout_\badtree =
\{\texttt{Error}, \bot\}
\cup
\{(\texttt{pointer},p) \mid p \in \{\Left,\Right,\chl,\chr\}\}$,
$\Eout_\badtree = \{\bot\}$.

\begin{framed}
\noindent\textbf{Constraints $\labels{C}{bad}{tree}$}
\begin{enumerate}
    \item
    A node $u$ may output $\texttt{Error}$ only if $u$
		violates at least one constraint of $\CC_\badtree$.

    \item
    If $u$ outputs $(\texttt{pointer},p)$,
    then there must exist an incident edge $e$ with $L_u(e)=p$.

    \item
    Let $v=f(u,p)$. Then $v$ must output either $\texttt{Error}$ or $(\texttt{pointer},p')$.
    In the latter case, the following must hold:
    \begin{enumerate}
        \item If $p \in \{\Left,\Right\}$, then $p' = p$.
        \item If $p = \Parent$, then $p' \in \{\Parent,\Left,\Right\}$.
        \item If $p = \chr$, then $p' \in \{\chr,\Left,\Right\}$.
    \end{enumerate}
\end{enumerate}
\end{framed}

We lift the following lemma stating that in perfect tree-like structures, the algorithm cannot convince a verifier that the input is invalid, from previous work:
\begin{lemma}[\cite{balliu-ghaffari-etal-2025-shared-randomness-helps-with}, Lemma~4.4]
	Let $G$ be a $(\Vin_\tree,\Ein_\tree)$-labeled input graph satisfying the constraints of  $\CC_\goodtree$. Then the only possible output for each node is $\bot$.
\end{lemma}

The following lemma also exists in the literature, but the previous proof assumes that the algorithm knows an upper bound on $n$.
Here we prove that this works even without any knowledge on the size of the graph:
\begin{lemma}\label{tree_certification_algm}
	Let $G$ be a  $(\Vin_\tree,\Ein_\tree)$-labeled connected input graph. Then there exists an $O(\log n)$-round \local algorithm for $\Pi_\tree^\bad$ such that 
	\begin{enumerate}[noitemsep]
		\item If $G$ is perfect tree-like, then every node outputs $\bot$.
		\item Otherwise, every node produces output other than $\bot$.
	\end{enumerate}
\end{lemma}

\begin{proof}
	Here is a \local certification algorithm which runs in $O(\log n)$ rounds.

	\paragraph{Certification algorithm.} Each node $v$ maintains the explored ball $N_G^r(v)$ after round $r$. In round $r$, for every newly discovered node $x\in N_G^r(v)\setminus N_G^{r-1}(v)$, node $v$ locally checks whether $x$ violates any constraint of $\CC_\goodtree$. If such a node is found, then $v$ starts outputting either $\texttt{Error}$ (if $v$ itself violates
a good-constraint) or a pointer $(\texttt{pointer},p)$ toward an error witness, following the allowed transition rules of $\CC_\badtree$. If no violation has been found so far, $v$ continues expanding until it can identify both the leftmost and the rightmost leaf of the (candidate) tree-like structure; in that case it outputs~$\bot$.

The correctness of the algorithm follows from the following two properties.

\begin{enumerate}
	\item If there is a node $u$ which doesn't satisfy the constraints $\CC_\goodtree$, then there is one at a  distance at most $O(\log n)$ from $v$.
	\item If node $v$ has seen both the leftmost and the rightmost leaf node in the tree (which are locally recognizable) without encountering any error node, then it has explored the whole graph $G$.
\end{enumerate}

	The first item follows from the fact that the perfect tree has all nodes with degree at least 3 (except the root node) and therefore, has diameter $O(\log n)$. Therefore, either node $v$ encounters the error node within distance $O(\log n)$ or it exhausts all nodes of $G$.

	For the second item, w.l.o.g., let us assume that the rightmost leaf node ($y$) is the farthest from $v$ among the two extreme leaf nodes. For every other node $x$ in $G$, $\dist_G(x,v) \leq \dist_G(x,r) +\dist_G(r,v)\leq \dist(v,r) + \dist(r,y) $ where $r$ is the root of the perfect tree. The inequality follows from the fact that $y$ is farthest from $r$ among all other nodes of $G$ (as it is a leaf node). Therefore, if $v$ encounters $y$ (without encountering any node with errors), then it has essentially explored the whole $G$ without finding any node which violate the constraints $\CC_\goodtree$ and therefore, can safely output $\bot$. Since the diameter of $G$ is at most $O(\log n)$, the certification algorithm takes at most $O(\log n)$ rounds.
\end{proof}

\subsubsection{Growing Grids}
Next, we show how we can use perfect trees to construct and certify a growing grid with the following properties:
\begin{itemize}[noitemsep]
    \item each node has a globally consistent orientation,
	\item each row is exactly one longer than the row to the north of it, and
	\item the distance between any two nodes on the same row is at most $O(\log n)$ using an attached perfect tree.
\end{itemize}

We certify this structure in two parts.
First, we certify that each row of the grid is proper and has an attached tree.
We then prove that these rows are attached together to form a growing grid.
\begin{definition}[Growing grid]
    \label{def:growing-grid}
	A graph $G = (V,E)$ is a growing grid of dimension $h$ and $\ell$ for $h,\ell\geq 1$ with vertex and edge sets defined as follows.
	\begin{enumerate}[noitemsep]
		\item $V = \{(i,j)\ \mid\ i\in [h], j\in [\ell+i-1]\}$
		\item $E_{\text{rows}} = \bigl\{\{(i,j),(i,j+1)\}\ \mid\ i\in [h], j\in [\ell+i-1]\bigr\}$
		\item $E_{\text{col}} = \bigl\{\{(i,j),(i+1,j)\}\ \mid\ i\in[h-1], j\in [\ell+i-1]\bigr\}$
		\item $E = E_{\text{rows}} \cup E_{\text{col}}$
        \end{enumerate}
\end{definition}

\noindent
\paragraph{LCL for rows.} $\prob{rows}{} = (\prob{rows}{good}, \prob{rows}{bad})$

\paragraph{LCL for good rows.}

\[
	\prob{rows}{good} = (\labels{V}{in}{row},\labels{E}{in}{row}, \labels{V}{out}{goodrow},\labels{E}{out}{goodrow}, \labels{C}{rows}{good}) 
\]
where 
$\labels{V}{in}{good} = \{\mathrm{T},\mathrm{G}\}$
$\labels{E}{in}{good} = \{\Parent, \ch, \east,\west\}$
$\labels{V}{out}{good} = \{\bot\}$
$\labels{E}{out}{good} = \{\bot\}$

If a node $v$ with label $\mathrm{G}$ has a node $u = f(u, \Parent)$ with label $\mathrm{T}$, we call node $u$ the \emph{parent of $v$}.
This way we treat the rows almost like they were layers of the tree, just that their width is not a power of two.

\begin{framed}
	\noindent\textbf{Constraints $\labels{C}{goodrows}{}$}

	\noindent
	Nodes which output $\bot$ must satisfy the following constraints.
    \begin{enumerate}
	    \item We assume that all nodes labeled $\mathrm{T}$ satisfy $\labels{C}{tree}{good}$. 
	
	    \item Every node labeled $G$ must have at most one half-edge labeled $\east$ and at most one half-edge labeled $\west$ and exactly one half-edge labeled $\Parent$. 
	    
	    \item Every node labeled $T$ can have at most one half-edge labeled $\ch$.
		  
	    \item \textbf{Every grid node is adjacent to some tree leaf node.}\label{rule:grid_adj_tree}
		
		    A node $u$ with label $\mathrm{G}$ must have an adjacent edge $e = \{u, v\}$ with $L_u(e) = \Parent$, $L_v(e) = \ch$, and node $v$ must have label $\mathrm{T}$. Moreover, $v$ must not have any half-edge incident on it labeled $\chl$ or $\chr$.
	    \item \textbf{A grid node adjacent to a tree node having a right neighbor must have an east neighbor. Moreover the grid node should participate in a $\mathbf 4$-cycle or a $\mathbf 5$-cycle with the adjacent tree.}\label{rule:left-east}
		
		A node $u$ with label $\mathrm{G}$ and a neighbor $v$ with an edge $e$ labeled with $L_v(e) = \Right$ must have an edge labeled with $\east$.
        Moreover, either $f(u, \Parent, \Right, \ch, \west) = u$ or $f(u, \Parent, \Right, \Right, \ch, \west) = u$.
		    
	    \item A node $u$ with label $\mathrm{G}$ has a neighboring edge with label $\west$ iff its parent has a neighboring edge with label $\Left$.
        \item A node $u$ with label $\mathrm{G}$ has a neighboring edge with label $\east$ iff its parent has a neighboring edge with label $\Right$.
    \end{enumerate}
\end{framed}

\paragraph{LCL for bad rows.}
\[
	\prob{rows}{bad} = (\labels{V}{in}{row},\labels{E}{in}{row}, \labels{V}{out}{badrow},\labels{E}{out}{badrow}, \labels{C}{rows}{bad}) 
\]

where 
$\labels{V}{out}{bad} = \{\texttt{error}\}\cup \{(\texttt{pointer},p)\ \mid\ p\in\{\east, \west\}\}$

\begin{framed}
	\noindent\textbf{Constraints $\labels{C}{badrows}{}$}
	\begin{enumerate}
		\item A node $u$ outputs \texttt{error} only if $u$ violates at least one constraint of $\labels{C}{badrows}{}$.
	        \item If $u$ outputs $(\texttt{pointer},p)$, then there must exist an incident edge $e$ with $L_u(e)=p$.
		\item
    Let $v=f(u,p)$. Then $v$ must output either $\texttt{Error}$ or $(\texttt{pointer},p')$.
    In the latter case, the following must hold: if $p \in \{\east,\west\}$, then $p' = p$.
    \end{enumerate}
\end{framed}

\begin{lemma}\label{lem:log-diameter-of-grid-tree}
	Let $G$ be a $(\labels{V}{tree}{in}\times \labels{V}{rows}{in}, \labels{E}{tree}{in}\times \labels{E}{rows}{in})$-labeled connected graph which satisfies the constraints of $\Pi_\tree\otimes \Pi_\row$. Then, it has diameter at most $O(\log n)$. 	
\end{lemma}

\begin{proof}
	As the nodes labeled $\mathrm{T}$ satisfy the constraints of $\CC_\goodtree$, we have that each graph induced on the tree nodes forms a perfect tree-like structure. We claim that there is exactly one tree-like structure in $G$. Let us assume that there are two sets of nodes labeled $\mathrm{T}$ such that they induce two disjoint trees $T_1$ and $T_2$. Rule~\ref{rule:left-east} of $\CC_\goodrow$ ensures that the nodes labeled $\mathrm{G}$ which are connected to $T_1$ and $T_2$ respectively form paths (let the paths be $P_1$ and $P_2$ respectively). Moreover, the endpoints of $P_1$ must be adjacent to the leftmost and rightmost leaf nodes of $T_1$ (and same for $T_2$ and $P_2$). Since we assume that $G$ is connected and $T_1$ and $T_2$ induce disjoint graphs, the only possible connection between $T_1\cup P_1$ and $T_2\cup P_2$ is via the endpoints of $P_1$ and $P_2$ (otherwise rule~2 or rule~4 is violated). But then rule~5 or rule~6 is violated (depending on whether the  edge is labeled $\Left$ or $\Right$ on the nodes that it is incident on) which contradicts our assumption. As the nodes of $\mathrm{T}$ are a part of the same tree-like structure and every node labeled $G$ is connected to a node labeled $T$, the diameter of $G$ is $O(\log n)$.          
\end{proof}

\begin{lemma}\label{alg:certify_rows}
	Let $G$ be a $(\labels{V}{tree}{in}\times \labels{V}{rows}{in}, \labels{E}{tree}{in}\times \labels{E}{rows}{in})$-labeled graph. There exists an $O(\log n)$-round \local algorithm solving $\prob{}{tree}\otimes \prob{}{row}$ such that 
	\begin{enumerate}
		\item If all the tree nodes satisfy $\labels{C}{goodtree}{}$ then they output $\bot$. Else every tree node produces  output other than $\bot$.
		\item If all grid nodes satisfy $\labels{C}{goodrows}{}$ (which implicitly conditions on the fact that every tree node outputs $\bot$), then they output $\bot$. Else every grid node produces output other than $\bot$.
	\end{enumerate}
\end{lemma}

\begin{proof}
	The algorithm proceeds as follows at a node $v\in V(G)$. If $v$ is labeled $\mathrm{T}$, then the algorithm executes exactly as in Lemma~\ref{tree_certification_algm}. Else if the node $v$ is labeled $\mathrm{G}$, then by rule~\ref{rule:grid_adj_tree}, it must be adjacent to a node labeled $\mathrm{T}$ (else it outputs \texttt{error}). It keeps querying its tree neighbor to check its output. If the tree-node  outputs $\bot$, then the tree is perfect-treelike. It then continues to grow its neighborhood till it either encounters a row-node which doesn't satisfy the constraints $\CC_\goodrow$ or it encounters the eastmost and the westmost node of the row. Since every grid node is adjacent to a tree-node of the same tree (refer to proof of Lemma~\ref{lem:log-diameter-of-grid-tree}), by exactly the same arguments as in Lemma~\ref{tree_certification_algm}, we have that in $O(\log n)$ rounds, either $v$ encounters an error grid node (and it outputs $(\texttt{pointer},p)$) following the convention of $\CC_\goodrow$ or it encounters  both the extreme nodes of the row, which then implies that it has explored the whole row which is free of erroneous nodes and therefore, it outputs~$\bot$.    
\end{proof}

\noindent
\paragraph{LCL for grids.}

Formally, the LCL for growing grid (with tree-like structure attached on each row) is defined as $\Pi_\grid=(\Pi_\grid^\good,\Pi_\grid^\bad)$. 
\paragraph{Good grid.}
\[
	\prob{grid}{good} =(\labels{V}{grid}{in}, \labels{E}{grid}{in},\labels{V}{goodgrid}{out},\labels{E}{goodgrid}{out}, \labels{C}{goodgrid}{})
\]
where $\labels{V}{in}{grid} = \{G\}$,
$\Ein_\grid = \{\west,\east,\north, \south\}\cup \labels{E}{in}{tree}$,
$\Vout_\grid = \{\bot\}$,
$\Eout_\grid = \{\bot\}$

\begin{framed}
	\noindent\textbf{Constraints $\CC_\goodgrid$}
	
    \noindent
    Nodes with output label $\bot$ must follow the following constraints.
    \begin{enumerate}
	 \item We assume that all nodes labeled $T$ satisfy the $\CC_\goodtree$ and all nodes labeled $G$ satisfy $\CC_\goodrow$. 
	 \item \textbf{Every grid node, except the right-most nodes, must have both $\north$ and $\south$ neighbors}.
        
		 If node $u$ has a neighboring edge $e = \{u, v\}$ with label $L_u(e) = \east$, then it must have a neighboring edge with label $\south$ exactly if $v$ has a neighboring edge with label $\south$.
        This ensures that the south-most row of the grid is consistent.
	\item If node $v$ has an adjacent half-edge labeled $\east$, then $f(v, \east, \south, \west, \north) = v$.\label{rule:cycle}
	\item If node $v$ doesn't have an adjacent half-edge labeled $\east$, then it doesn't have an adjacent half-edge labeled with $\north$.
	
		This ensures that the graph looks like a grid and expands by at least one each row.
    \end{enumerate}
\end{framed}

\paragraph{Bad grid.} The LCL $\Pi^{\badgrid}$ is defined as follows.

\[
	\Pi^{\badgrid} = (\Vin_\badgrid,\Ein_\badgrid, \Vout_\badgrid, \Eout_\badgrid,\CC_\badgrid)
\]

where  $\Eout_\badgrid = \{\bot\}$ and

$\Vout_\grid = \Vout_\badtree\cup \{(\texttt{pointer},p)\ \mid\ p\in \{\east,\west,\Parent,\ch\}\}$

\begin{framed}
	\noindent\textbf{Constraints $\CC_\badgrid$}
	\begin{enumerate}
		\item A node $u$ outputs \texttt{error} only if $u$ violates at least one constraint of $\labels{C}{badrows}{}$.
	        \item If $u$ outputs $(\texttt{pointer},p)$, then there must exist an incident edge $e$ with $L_u(e)=p$.
		\item
    Let $v=f(u,p)$. Then $v$ must output either $\texttt{Error}$ or $(\texttt{pointer},p')$.
    In the latter case, the following must hold: if $p \in \{\east,\west\}$, then $p' = p$.
    \end{enumerate}
\end{framed}

\begin{lemma}
    \label{lem:grid-algorithm}
	There exists an $O(\log n)$-round labeling algorithm which labels an $n$-vertex graph $G$ such that the following holds.
	\begin{enumerate}
		\item The algorithm solves $\Pi_\tree\otimes\Pi_\row\otimes\Pi_\grid$
		\item Nodes satisfying $\Pi_\goodgrid$ form a valid growing grid.
	\end{enumerate}
\end{lemma}

\begin{proof}
	The algorithm proceeds exactly as in Lemma~\ref{alg:certify_rows}. The algorithm proceeds as follows at a node $v\in V(G)$. If $v$ is labeled $\mathrm{T}$, then the algorithm executes Lemma~\ref{tree_certification_algm}. Else if the node $v$ is labeled $\mathrm{G}$, then it must be adjacent to a node labeled $\mathrm{T}$ (else it outputs \texttt{error}). It keeps querying its tree neighbor to check its output. If the tree-node  outputs $\bot$, then the tree is perfect-treelike. It then continues to grow its neighborhood till it either encounters a row-node which doesn't satisfy the constraints $\CC_\goodrow$ or $\CC_\goodgrid$ or it encounters the eastmost and the westmost node of the row. The fact that the algorithm terminates in $O(\log n)$ rounds follows directly from Lemma~\ref{alg:certify_rows}.  Every tree node accepts $\bot$ if and only if it satisfies $\CC_{\goodgrid}$ otherwise it produces an output satisfying the convention of $\CC_\badtree$. Every grid node outputs $\bot$ if and only if it satisfies the constraints of $\CC_\goodgrid$ and $\CC_\goodrow$. Otherwise it produces an output following the convention of $\CC_\badrow$ or $\CC_\badgrid$ (depending on which set of constraints are violated). Hence the algorithm solves $\Pi_\tree\otimes\Pi_\row\otimes\Pi_\grid$.

To prove the second item, let us assume that the set of grid nodes satisfy $\Pi_\goodgrid$. W.l.o.g., let us assume that they induce a connected graph. We start by labeling each node with a label from $\mathbb{N}\times \mathbb{N}$ inductively as follows:

	\begin{enumerate}
		\item The node which doesn't have a north or west neighbor is labeled $(1,1)$.
		\item For every node $u$ labeled $(i,j)$, if it satisfies rule~\ref{rule:cycle}, then for nodes $v_1,\ v_2,\ v_3$, such that $f(u,\east) = v_1$,$f(u,\east,\south) = v_2$ and $f(u,\east, \south, \west) = v_3$, the labeling is as follows: $v_1$ is labeled $(i,j+1)$, $v_2$ is labeled $(i+1,j+1)$ and  $v_3$ is labeled $(i+1,j)$.
	\end{enumerate}

Let the maximum value for the first and the second arguments in the tuple $(i,j)$ in the above labeling be $h$ and $\ell$ respectively.

We claim that for every $i\in [h]$ and $j\in [\ell+i-1]$, the following holds:
\begin{enumerate}[noitemsep]
        \item if $j<\ell+i-1$, then $\{(i,j),(i,j+1)\}\in E(G)$;
        \item if $i<h$, then $\{(i,j),(i+1,j)\}\in E(G)$.
\end{enumerate}

	We prove this via a double induction on $(i,j)$.

	\begin{enumerate}
		\item \textbf{Base $(1,1)$.} If $\ell>1$, then $(1,1)$ is not the right-most node of its row, hence it has an $\east$-neighbor by $\CC_\goodrow$, and therefore $\{(1,1),(1,2)\}\in E(G)$.
If $h>1$ and $\ell>1$, then $(1,1)$ is not right-most, so by rule~2 of $\CC_\goodgrid$ it has a $\south$-neighbor, and hence $\{(1,1),(2,1)\}\in E(G)$.

 	 \item \textbf{Inductive step.}
Assume the statement holds for all $(i',j')\prec(i,j)$.

\begin{enumerate}[noitemsep]
        \item If $j<\ell+i-1$, then $(i,j)$ is not right-most, hence it has an $\east$-half-edge.
        Applying rule~\ref{rule:cycle} at $(i,j)$ yields $f((i,j),\east)=(i,j+1)$, and therefore
        $\{(i,j),(i,j+1)\}\in E(G)$.

        \item If $i<h$, we show that $\{(i,j),(i+1,j)\}\in E(G)$.
        If $j<\ell+i-1$, then $(i,j)$ is not right-most, and by rule~2 of $\CC_\goodgrid$ it has a $\south$-neighbor;
        hence $\{(i,j),(i+1,j)\}\in E(G)$.
        Otherwise $j=\ell+i-1$ (the right-most node). Then $(i,j-1)\prec(i,j)$ and $j-1<\ell+i-1$, so by the inductive
        hypothesis $\{(i,j-1),(i,j)\}\in E(G)$ and thus $(i,j-1)$ has an $\east$-half-edge.
        Applying rule~\ref{rule:cycle} at $(i,j-1)$ forces the completing $4$-cycle, and in particular yields the vertical
        edge $\{(i,j),(i+1,j)\}$.
\end{enumerate}
	\end{enumerate}

	Therefore, the induced subgraph on the grid nodes contains exactly the edges in $E_{\text{rows}}\cup E_{\text{col}}$ that are contained in a growing grid of dimension $(h,\ell)$. 
\end{proof}

\subsubsection{Turing Machine Execution}

We now move on to defining a Turing machine on the growing grid.
We start by fixing our definition of Turing machines so that a machine \(M\) is a tuple \((Q,\Sigma,q_0,F,\delta)\) where:
\begin{itemize}[noitemsep]
    \item \(Q\) is a set of states;
    \item \(\Sigma\) is the set of all symbols that can appear on the tape;
    \item $F = \{ \accept, \reject \}$ is the set of halting states (where of
    course $\accept \neq \reject$);
    \item \(\delta: Q\times\Sigma\rightarrow Q\times\Sigma\times\{\mathtt{L},\mathtt{R},\mathtt{S}\}\) is the transition function that describes how a single computation step is performed.
    If the machine is either in state $\accept$ or $\reject$, then the output Turing machine state must not evolve, that is
    \[
        \delta(q, b) = (q, b, \mathtt{S}) \quad \text{for $q \in \{\accept, \reject\}$.}
    \]
\end{itemize}
Note that we consider Turing machines that explicitly either accept or reject the input. With the goal of maintaining a promise free approach, we adopt the following scheme: First, we define constraints on the input labels of \(\Pi_\turing\) that certify whether the input labeling encodes a valid computation of a universal Turing machine. If this is not the case, then we define constraints that make the output labels point towards a mistake inside a row of the grid. This essentially partitions the graph into regions of valid computations, separated by regions with errors. Then, within the region of valid computations, the output node labels have to satisfy constraints that are dependent on the final state of the Turing-machine computation. In effect, the output \emph{predicts} whether the Turing machine halts in an accepting or a rejecting state. This final step is the key ingredient to show a separation between computable and uncomputable \local.

\paragraph{Turing LCL.}
Here we define the LCL problem $\Pi_\turing$, which will allow us to show the separation between computable and uncomputable \local. We start by picking a universal Turing machine $U$ on binary alphabet, and we denote $\descr{M}$ the description of Turing machine $M$ on the tape of $U$. Note that we aim to encode the computation of a universal Turing machine \(U\). This allows us to keep the size of its representation constant without restricting expressive power. Otherwise, the transition function would change as we try to encode a different Turing machine, and we would not be able to define an LCL. 
We define a certification LCL $\Pi_\turing = \bigl(\Pi^\good_\turing, \Pi^\bad_\turing\bigr)$ as follows:
Let \(\Sigma_{\operatorname{tape}}=\{0,1\}\), \(\Sigma_{\operatorname{head}}=\{H,\botTM\}\) and \(\Sigma_{\operatorname{state}}=Q\cup\{\botTM\}\).
Then, the good and bad Turing LCLs are defined as follows:
\begin{align*}
  \Pi^\good_\turing &= \bigl(\Vin_\turing,\Ein_\turing,\Vout_\goodturing,\Eout_\goodturing,\CC_\goodturing\bigr), \text{ and} \\
  \Pi^\bad_\turing &= \bigl(\Vin_\turing,\Ein_\turing,\Vout_\badturing,\Eout_\badturing,\CC_{\badturing}\bigr),
\end{align*}
where
\begin{itemize}[noitemsep]
  \item \(\Vin_\turing = \Sigma_{\operatorname{tape}}\times\Sigma_{\operatorname{head}}\times\Sigma_{\operatorname{state}}\)
  \item \(\Ein_\turing=\{I_H,O_H\}\)
  \item \(\Vout_\goodturing=\{\bot\}\)
  \item \(\Eout_\goodturing=\{\bot\}\)
  \item \(\Vout_\badturing=\{\bot\}\)
  \item \(\Eout_\badturing=\{I_\mathtt{E},O_\mathtt{E}\}\)
\end{itemize}

We start by defining \(\CC_\goodturing\), which certifies that the input labelings for nodes and edges correctly represent the computation of a universal Turing-machine \(U\). To give an intuition, Constraints~\ref{item:constGoodTuring-1a}--\ref{item:constGoodTuring-1e} use pointers to guarantee that the head of the Turing machine is unique, Constraint~\ref{item:constGoodTuring-2} enforces a proper formatting of the tape, and finally Constraints~\ref{item:constGoodTuring-5a}--\ref{item:constGoodTuring-5c} and~\ref{item:constGoodTuring-6} check the correctness of the transition function. Wherever we use labels \(\north\), \(\south\), \(\east\) and \(\west\) that are not in \(\Ein\), these refer to the labels used to encode the grid in \(\CC_\goodgrid\).

\begin{framed}
  \noindent\textbf{Constraints \(\CC_\goodturing\)}

  \noindent
  Let node \(u\) have input \((b,h,q)\in\Sigma_{\operatorname{tape}}\times\Sigma_{\operatorname{head}}\times\Sigma_{\operatorname{state}}\) and let \(w=f(u,\west)\), \(v=f(u,\east)\), \(e_w=(u,w)\) and \(e_v=(u,v)\), if they exist.
  \begin{enumerate}

    \item \textbf{The head is unique}
      \begin{enumerate}
      \item \label{item:constGoodTuring-1a}
      \textbf{The neighbors point toward the head}\\
      If \(h=H\), then \(L^\text{in}_u(e_w)=I_H\), \(L^\text{in}_w(e_w)=O_H\), \(L^\text{in}_u(e_v)=I_H\) and \(L^\text{in}_v(e_v)=O_H\).

      \item \label{item:constGoodTuring-1b}
      \textbf{The leftmost node points right if it is not the head}\\
      If \(h\neq H\) and \(w\) does not exist, then \(L^\text{in}_u(e_v)=O_H\) and \(L^\text{in}_v(e_v)=I_H\).

      \item \label{item:constGoodTuring-1c}
      \textbf{The rightmost node points left if it is not the head}\\
      If \(h\neq H\) and \(v\) does not exist, then \(L^\text{in}_u(e_w)=O_H\) and \(L^\text{in}_w(e_w)=I_H\).

      \item \label{item:constGoodTuring-1d}
      \textbf{The chain of pointers is propagated right}\\
      If \(h\neq H\) and \(L^\text{in}_w(e_w)=I_H\), then \(L^\text{in}_u(e_w)=O_H\) and \(L^\text{in}_u(e_v)=I_H\).

      \item \label{item:constGoodTuring-1e}
      \textbf{The chain of pointers is propagated left}\\
      If \(h\neq H\) and \(L^\text{in}_v(e_v)=I_H\), then \(L^\text{in}_u(e_v)=O_H\) and \(L^\text{in}_u(e_w)=I_H\).
    \end{enumerate}

    \item \label{item:constGoodTuring-2}
    \textbf{Only the node with the head has its state label different from \(\botTM\)}\\
    \(q \in Q\) if and only if \(h=H\) (otherwise \(q=\botTM\)).

    \item \textbf{Computation is valid}\\
    If \(h=H\), then \(\delta(b,q)=(b_\next,q_\next,d)\) and moreover:
    \begin{enumerate}
      \item \label{item:constGoodTuring-5a}
      \textbf{The head does not move}\\
      If $d = \mathtt{S}$ and node $f(u,\south)$ exists, then it must have input $(b_\next, H, q_\next)$.

      \item \label{item:constGoodTuring-5b}
      \textbf{The head moves right}\\
      If $d = \mathtt{R}$ and node $f(u,\south,\east)$ exists, then it must have input $(b_{\south\east}, H, q_\next)$ for some value $b_{\south\east}$.
      Also, node $f(u,\south)$ must have input $(b_\next, \botTM, \botTM)$.

      \item \label{item:constGoodTuring-5c}
      \textbf{The head moves left}\\
      If $d = \mathtt{L}$ and node $f(u,\south,\west)$ exists, then it must have input $(b_{\south\west}, H, q_\next)$ for some value $b_{\south\west}$.
      Also, node $f(u,\south)$ must have input $(b_\next, \botTM, \botTM)$.
    \end{enumerate}

    \item \label{item:constGoodTuring-6}
    \textbf{The symbols on the rest of the tape are copied to the next row}\\
    If \(h\neq H\) and \(f(u,\south)\) exists, then it must have input \((b,h_\south,q_\south)\) for some values $h_\south$ and $q_\south$.
  \end{enumerate}
\end{framed}

Next, we define \(\CC_\badturing\). These constraints use pointers to invalidate an entire row of the grid when at least one node detects an error in the input labeling.

\begin{framed}
  \noindent\textbf{Constraints \(\CC_\badturing\)}

  \noindent
  Let node \(u\) have input \((b,h,q)\in\Sigma_{\operatorname{tape}}\times\Sigma_{\operatorname{head}}\times\Sigma_{\operatorname{state}}\) and let \(w=f(u,\west)\), \(v=f(u,\east)\), \(e_w=(u,w)\) and \(e_v=(u,v)\), if they exist.
  \begin{enumerate}
    \item \label{item:constBadTuring-edges}
    \textbf{Outgoing edges have a matching incoming end}\\
    For edge $e_w$, it must hold that if $\Lout_u(e_w) = O_\mathtt{E}$, then $\Lout_w(e_w) = I_\mathtt{E}$.
    Note that an edge may have $I_\mathtt{E}$ at both ends.

    \item \label{item:constBadTuring-sink}
    \textbf{The neighbors point toward a node with an error}\\
    If \(u\) does not respect at least one constraint in \(\CC_\goodturing\), then \(\Lout_u(e_w)=I_\mathtt{E}\) and \(\Lout_u(e_v)=I_\mathtt{E}\), if they exist.

    \item \label{item:constBadTuring-propagate}
    \textbf{The chain of error pointers is propagated}\\
    If $u$ satisfies $\CC_\goodturing$, then either $\Lout_u(e_w) = O_\mathtt{E}$ or $\Lout_u(e_v) = O_\mathtt{E}$.
  \end{enumerate}
\end{framed}

We now prove that solving \(\Pi_\turing\) means partitioning the rows of the grid in regions with errors and regions where we encode a proper computation of a Turing-machine.
In the following lemmas, we refer to \(G\), a \((\Vi,\Ei)\)-labeled input graph satisfying the constraints \(\CC_\goodgrid\).
We also use the notation \(r=(u_1,\ldots,u_m)\) to denote a row of \(G\), that is a path in \(G\) such that for every \(e_i=(u_i,u_{i+1})\), where \(i\in[1,m-1]\), it holds that \(\Lin_{u_i}(e_i)=\east\) and \(\Lin_{u_{i+1}}(e_i)=\west\), and also \(u_m\) has no half-edge labeled \(\east\) on its side.
Moreover, we remind that \(U\) is the universal Turing machine that we encoded in our constraints, and thus let us give the following definition.
\begin{definition}
  Let \((b_i,h_i,q_i)\) be the input for node \(u_i\) in a row \(r=(u_1,\ldots,u_m)\) of \(G\). Row \(r\) \emph{represents} a tape of \(U\) if the following hold:
  \begin{enumerate}[noitemsep]
    \item\label{def:item:tape-representation-1} there exists exactly one \(i\) such that \(h_i=H\);
    \item\label{def:item:tape-representation-2} for every \(i\in[1,m]\), if \(h_i=H\) then \(q_i\in Q\), otherwise \(q_i=\botTM\)
  \end{enumerate}
\end{definition}

Then, constraints \(\CC_\goodturing\) ensure the following properties on the rows of \(G\). First, we ensure we can encode a proper representation of the tape.
\begin{lemma}
\label{lem:row-is-TM-tape}
  Let \(r=(u_1,\ldots,u_m)\) be a row of \(G\). If every node in \(r\) respects constraints \(\CC_\goodturing\), then \(r\) represents a tape of \(U\).
\end{lemma}
\begin{proof}
  We analyze case by case the conditions that must be satisfied to obtain a valid tape. Condition~\ref{def:item:tape-representation-2} is trivially guaranteed by Constraint~\ref{item:constGoodTuring-2} in \(\CC_\goodturing\).
  Condition~\ref{def:item:tape-representation-1} is guaranteed by Constraints~\ref{item:constGoodTuring-1a}-\ref{item:constGoodTuring-1e} in the following way. If \(h_i=H\), then Constraint~\ref{item:constGoodTuring-1a} forces the neighbors of \(u_i\) to point towards it, and this is the case even when \(u_i\) is the leftmost node \(u_1\) or the rightmost one \(u_m\). The key property here is that this is the only case in which two \(I_H\) labels can be on the half-edges of the same node. Then, if \(h_1\neq H\), Constraint~\ref{item:constGoodTuring-1b} starts a chain of pointers at \(u_1\) pointing right, and this is propagated by Constraint~\ref{item:constGoodTuring-1d}. Similarly, if \(h_m\neq H\), Constraint~\ref{item:constGoodTuring-1c} starts a chain of pointers at \(u_m\) pointing left, and this is propagated by Constraint~\ref{item:constGoodTuring-1e}. The only way to reconcile these two chains of pointers is to have a single node \(u_i\) with \(h_i=H\), as it is able to have label \(I_H\) on both half-edges.
  
  To see why we cannot have multiple heads, assume there are two nodes \(u_i\) and \(u_j\) where \(h_i=h_j=H\) for \(i<j\). If the two nodes are next to each other, that is \(i+1=j\), Constraint~\ref{item:constGoodTuring-1a} is violated. Otherwise, Constraints~\ref{item:constGoodTuring-1d} and~\ref{item:constGoodTuring-1e} force either the existence of a node \(u_k\) between \(u_i\) and \(u_j\) with both half-edges labelled \(O_H\), or the existence of an edge with both labels \(I_H\), and this contradicts the very Constraints~\ref{item:constGoodTuring-1d} and~\ref{item:constGoodTuring-1e} themselves.
\end{proof}

To prove that the computation is correct, we ensure that we encode one application of the transition function correctly.
\begin{lemma}
\label{lem:TM-computation-step}
  Let \(r_1\) and \(r_2\) be consecutive rows of \(G\). If every node in both \(r_1\) and \(r_2\) respects constraints \(\CC_\goodturing\), then \(r_2\) represents the tape obtained by applying the transition function \(\delta\) of universal Turing machine \(U\) to the tape represented by \(r_1\).
\end{lemma}
\begin{proof}
  Thanks to \Cref{lem:row-is-TM-tape}, we know that both \(r_1\) and \(r_2\) represent valid Turing machine tapes. Let \(u_1\) be a node in \(r_1\) with input \(b_1,h_1,q_1\) and let \(\delta(b,q)=(b_\next,q_\next,d)\), where \(\delta\) is the transition function of universal Turing machine \(U\). We consider two cases.

  \textbf{Case \(h_1=H\).} The new character in this position of the tape must be \(b_\next\), and all the three Constraints~\ref{item:constGoodTuring-5a}, \ref{item:constGoodTuring-5b} and~\ref{item:constGoodTuring-5c} ensure this. If the head does not move, Constraint~\ref{item:constGoodTuring-5a} guarantees that the node south to \(u_1\) must hold \(b_\next\) and \(q_\next\). If the head moves right, Constraint~\ref{item:constGoodTuring-5b} forces the new symbol to be written in the current position of the tape by asking the node south to \(u_1\) to have input \((b_\next,\botTM,\botTM)\). Constraint~\ref{item:constGoodTuring-5b} also forces the node south-east to \(u_1\) to have input \((b_{\south\east},H,q_{\south\east})\), which means the head is moving to the right. An analogous behavior is enforced by Constraint~\ref{item:constGoodTuring-5c} when the head moves left.

  \textbf{Case \(h_1\neq H\).} In this case, the only thing to ensure is that every tape symbol on row \(r_1\) is copied to the next row \(r_2\). This is guaranteed by Constraint~\ref{item:constGoodTuring-6}. Notice that this holds also for those nodes next to the head, and this constraint poses a condition only on the value of the symbol on the tape. Thus, this does not conflict with Constraints~\ref{item:constGoodTuring-5b} and~\ref{item:constGoodTuring-5c}, as they pose a condition only on the value of the head and the state.
\end{proof}

Finally, we prove that if the input encodes valid tapes of universal Turing machine $U$, then the algorithm cannot solve problem $\CC_\badturing$:
\begin{lemma}
\label{lem:TM-error-pointers}
  Let \(r_1=(u_1,\ldots,u_m)\) and \(r_2\) be consecutive rows of \(G\).
  If row $r_1$ represents a tape of $U$, then there exists an input labeling on the edges and on $r_2$ such that nodes of $r_1$ can satisfy $\CC_\goodturing$.
  Moreover, they cannot satisfy $\CC_\badturing$.
\end{lemma}
\begin{proof}
  If row $r_1$ represents a tape of $U$, then it suffices to make row $r_2$ encode the state of the tape after one step of the universal Turing machine $U$.
  This clearly satisfies Constraints~\ref{item:constGoodTuring-2}, \ref{item:constGoodTuring-5a}--\ref{item:constGoodTuring-5c} and \ref{item:constGoodTuring-6} of $\CC_\goodturing$.
  To satisfy Constraints~\ref{item:constGoodTuring-1a}--\ref{item:constGoodTuring-1e}, it suffices to orient the input edges on $r_1$ towards the head of the Turing machine.
  
  To show that none of the nodes may satisfy $\CC_\badturing$, it is clear that none of the nodes may satisfy Constraint~\ref{item:constBadTuring-sink} as they satisfy $\CC_\goodturing$.
  Assume for contradiction that at least one node of row $r_1$ satisfies $\CC_\badturing$.
  This immediately implies that all nodes of $r_1$ satisfy $\CC_\badturing$ as the output labels of $\Pi_\goodturing$ and $\Pi_\badturing$ are incompatible with each other.
  Hence all nodes must satisfy Constraint~\ref{item:constBadTuring-propagate}.
  However, this, along with Constraint~\ref{item:constBadTuring-edges}, enforces that all nodes must point towards a node not satisfying $\CC_\goodturing$ in a consistent manner.
  As Constraint~\ref{item:constBadTuring-propagate} also ensures that the pointers cannot propagate over the endpoints of the row, and none of the nodes satisfy Constraint~\ref{item:constBadTuring-sink}, it is evident that the labeling will be invalid at least somewhere on $r_1$.
\end{proof}

We have now shown that $\Pi_\turing$ indeed certifies that the input consists of fragments of traces of execution of universal Turing machine $U$:
\begin{corollary}
    \label{cor:turing-components}
    Let $G$ be a labeled graph satisfying $\Pi_\tree \otimes \Pi_\row \otimes \Pi_\grid \otimes \Pi_\goodturing$.
    Then it consists of components which encode a correct execution of the universal Turing machine $U$.
\end{corollary}
We also note that there exists a locality-$O(\log n)$ \clocal algorithm for certifying this structure:
\begin{lemma}
    \label{lem:turing-algorithm}
    There exists a locality-$O(\log n)$ \clocal algorithm which, given a labeled graph $G$,
    \begin{enumerate}[noitemsep]
        \item Solves $\Pi_\tree \otimes \Pi_\row \otimes \Pi_\grid \otimes \Pi_\turing$, and
        \item Those regions where the algorithm solves $\Pi_\goodturing$ form a valid execution trace of Turing machine $U$.
    \end{enumerate}
\end{lemma}
\begin{proof}
	By \cref{lem:grid-algorithm}, we can solve problem $\Pi_\tree \otimes \Pi_\row \otimes \Pi_\grid$ with locality $O(\log n)$.
	We can concentrate on components which solve $\Pi_\goodgrid$; those form a growing grid.
	Moreover, the nodes of every row are within distance $O(\log n)$ from each other via the trees attached to the rows.
	
	Now the algorithm is as follows:
	Every node $u$ collects the whole row $r_1$ it is on, and the row $r_2$ south of it; it can do this with locality $O(\log n)$ using the trees.
	It then checks whether all nodes of $r_1$ satisfy constraints $\CC_\goodturing$; if so, it outputs $\bot$.
	Otherwise it outputs labels $I_\mathtt{E}$ and $O_\mathtt{E}$ on its adjacent edges such that they point towards an error on row $r_1$.
	
	It is easy to see that either all nodes of a row output $\bot$, in which case the whole row solves $\Pi_\goodturing$, or all nodes output error pointers, in which case the whole row solves $\Pi_\badturing$.
\end{proof}

\subsubsection{Consensus Problem}

Finally, we define the consensus LCL problem as follows:
\[
  \Pi_\consensus = (\{\bot\},\{\bot\},\{\accept, \reject\},\{\bot\},\CC_\consensus).
\]

The constraint $\CC_\consensus$ simply imposes that those regions of the grid where a proper encoding of a Turing machine computation is found, the output of all nodes is consistent.
Moreover, if the Turing machine halts, then the output will coincide with the halting state of the Turing machine.
Here, like in the definition of $\CC_\goodturing$, input labels $\north$, $\south$, $\east$ and $\west$ on edges refer to the labels on $\Pi_\grid$, and input labels $\Sigma_{\operatorname{tape}}\times\Sigma_{\operatorname{head}}\times\Sigma_{\operatorname{state}}$ on nodes refer to input labels on $\Pi_\turing$.
\begin{framed}
  \noindent\textbf{Constraints \(\CC_\consensus\)}

  \noindent
  Let node \(u\) have input \((b,h,q)\in\Sigma_{\operatorname{tape}}\times\Sigma_{\operatorname{head}}\times\Sigma_{\operatorname{state}}\) and let \(w=f(u,\west)\), \(v=f(u,\east)\), \(e_w=(u,w)\) and \(e_v=(u,v)\), if they exist.
  \begin{enumerate}
    \item \label{item:constTuringConsensus-1}
    \textbf{The output of a node in a final state is the symbol on the tape}\\
    If \(q \in \{\accept, \reject\}\), then $u$ outputs \(q\).

    \item \label{item:constTuringConsensus-2}
    \textbf{Outputs must be the same for all nodes in the same row}\\
    If node \(u\) outputs $x$, then both nodes \(w\) and \(v\) must output $x$, if they exist.

    \item \label{item:constTuringConsensus-3}
    \textbf{Outputs propagate upwards in a column}\\
    If node \(u\) outputs $x$, then node \(f(u,\north)\) must output $x$, if it exists.
  \end{enumerate}
\end{framed}

\subsection{Separation}

The final LCL we defined in the last section is
\[
    \Pi = \Pi_\tree \otimes \Pi_\row \otimes \Pi_\grid \otimes \Pi_\turing \otimes \Pi_\consensus .
\]
Now that $\Pi$ is defined, we may present the actual separation result.
First we deal with the upper bound part.

\begin{lemma}
  There is an \ulocal algorithm that solves $\Pi$ with locality $O(\log
  n)$.
\end{lemma}

\begin{proof}
  Since we are designing an \ulocal algorithm, we may as well assume that
  we have an oracle $H$ that, given a configuration $c$ of our universal Turing
  machine $U$, returns the state $H(c)$ in which $U$ halts when starting from
  $c$ or some special symbol to indicate $U$ does not halt.
  
  By \cref{lem:turing-algorithm}, we have a \clocal algorithm that solves
  \[
    \Pi_\tree \otimes \Pi_\row \otimes \Pi_\grid \otimes \Pi_\turing
  \]
  with locality $O(\log n)$.
  
  It therefore suffices us to consider here only the components which satisfy $\Pi_\goodturing$; by \cref{cor:turing-components}, these components form a correct execution of the universal Turing machine $U$.

  As all nodes of each row are within distance $O(\log n)$ from each other with the help of the trees, we can operate each row as a single entity which is controlled by the west-most node of the row.
  From the input labels, the row can then determine the configuration $c$ that is encoded in the row and query $H$.
  After doing so, it sets the output for the nodes it controls to $b$, where
  \[
    b = \begin{cases}
      H(c), & \text{$U$ halts when starting from $c$} \\
      \reject, & \text{otherwise.}
    \end{cases}
  \]
  Since all rows use the same oracle $H$ and $H(c) = H(c')$ if $c'$ is a
  predecessor configuration of $c$ (and $U$ is deterministic), the output of all
  rows is consistent across the same connected component: if $U$ halts when
  starting from the very first configuration $c_0$ in the component, then all
  nodes output $H(c_0)$; otherwise, $U$ does not halt starting from $c_0$,
  and hence all nodes output $\reject$.
\end{proof}

The lower bound part of the separation is as follows:

\begin{lemma}
  \label{lem:separation-lb}
  Any \clocal algorithm that solves $\Pi$ requires $\Omega(\sqrt{n})$
  locality.
\end{lemma}

To prove this, we need the following elementary diagonalization strategy:

\begin{lemma}
  \label{lem:diagonalization}
  There is a Turing machine $D$ such that, for every Turing machine $T$, the
  following holds:
  \begin{enumerate}[noitemsep]
    \item If $T$ does not halt on the input $x_T = (\descr{D}, \descr{T})$, then
    $D$ also does not halt on $x_T$.
    \item Otherwise $D$ and $T$ disagree on $x_T$, that is, we have $D(x_T) \neq
    T(x_T)$.
  \end{enumerate}
\end{lemma}

\begin{proof}
  We describe the operation of a Turing machine $D$.
  Using standard techniques from computability theory (i.e., so-called
  \enquote{quines}; see, e.g., \cite{sipser13_introduction_book}), we may assume
  that $D$ has access to its own description.
  Given $x_T$, $D$ simply does the following:
  \begin{enumerate}[noitemsep]
    \item Simulate $T$ on $x_T$.
    \item If $T$ halts, then output $1 - T(x_T)$.
  \end{enumerate}
  Hence, if $T$ does not halt, then so will $D$ not halt.
  Otherwise $D$ outputs exactly the opposite of what $T$ outputs on $x_T$.
\end{proof}

As usual in computability theory, it might seem at first confusing what results
from executing $D$ on the input $x_D$; however, notice that the description is
consistent and this just results in an infinite recursion loop.

The roadmap for the proof is as follows:
First we define a sequence of instances $G_k$ for $k \in \N_+$, in particular
setting identifiers on all the nodes.
Then, assuming we have a \local algorithm $\mathcal{A}$ that runs in
$o(\sqrt{n})$ time and is computable, we can derive a Turing machine $T$ that
simulates $\mathcal{A}$ and outputs whatever it does.
Finally, using \cref{lem:diagonalization}, $D$ gives us an adversarial instance
that $\mathcal{A}$ fails to solve correctly.

\begin{proof}[Proof of \cref{lem:separation-lb}]
  Let us first fix a sequence of instances $(G_k)_{k \in \N_+}$.
  The specification of the instance $G_k$ is as follows:
  \begin{itemize}
    \item $G_k$ consists of a single connected component that contains a
    growing grid (with tree structures placed on top of it) over exactly $k$
    rows, where the first row is a single node.

    \item We refer to the bundle of a row together with the tree structure on
    top of it as a \emph{layer} of $G_k$.
    The layers are numbered from $1$ to $k$ from top to bottom.
    Hence the $i$-th layer has $i$ nodes in its row part.

    \item We place identifiers in increasing order of layer, following some
    arbitrary but fixed scheme to assign identifiers within each layer.
    Hence the largest identifier is identical to the number of nodes $n$ and
    also every node in layer $i$ has an identifier that is strictly less than
    the identifiers appearing in a subsequent layer $j > i$.

    \item The inputs are chosen so that all nodes, except for nodes on constantly-many layers close to layers $1$ and $k$, must solve $\Pi_\goodgrid$.
    In addition, we define a binary \emph{input} $x$ with length $m \le k$ to $G_k$.
    Let $t_x \in \N_0 \cup \{ \infty \}$ be the number of steps that $U$ runs
    for when given input $x$.
    Then $x$ defines the rest of the inputs to the nodes in $G_k$ as follows:
    \begin{itemize}
      \item For the row nodes that are in a layer preceding the $m$-th one, the
      input is $o$ throughout the row.
      \item For the row nodes in the $m$-th layer, place $x$ written from left
      to right on that row, with the Turing machine head placed on the first
      symbol of $x$.
      \item For the row nodes that are in a layer subsequent to the $m$-th one,
      we extend the configuration inductively starting from the $m$-th layer by
      applying the transition function specified by $U$.
      Hence the $m'$-th layer is the $(m'-m)$-th configuration that $U$ reaches
      starting from $m$.
      If $m' > t_x$, then we place the halting configuration across the row.
    \end{itemize}
  \end{itemize}
  Notice that $G_k$ has $\Theta(k^2)$ nodes.
  Furthermore, we have that $G_k$ and $G_{k'}$ are identical in the first
  $\min\{k,k'\}$ layers.

  We now fix a computable distributed algorithm $\mathcal{A}$ with locality $o(\sqrt{n})$.
  Observe that, together with how we constructed $G_k$, this implies we can make
  the following statement:
  For every given $m \in \N_+$, there is a \emph{constant} $k_0 = k_0(m) > m$
  such that, for any $k \ge k_0$, when we run $\mathcal{A}$ on $G_k$ with any
  input $x$ of our choosing, all the nodes in the $m$-th layer of $G_k$ stop
  their operation and produce an output within the first $k-m-1$ rounds.
  In particular, this means that the nodes in layer $m$ do not see farther than
  the first $k-1$ layers in $G_k$ even as $k$ goes to infinity.
  Thus the number of rounds it takes for these nodes to commit to an output is
  \emph{constant} and independent of $k$, at least for large enough values of
  $k$.
  For the sake of uniqueness, we fix $k_0$ to be the least such value for which
  this statement holds.

  In addition to the above, we assume that $\mathcal{A}$ always produces a
  consistent output to $\Pi_\turing$ in each layer (i.e., there are no
  disagreements concerning the output bit) since otherwise the claim is trivial.

  With this in place, we now define a Turing machine $T$ as follows:
  \begin{enumerate}
    \item Given an input $x$ of length $m$, iterate over $k$ to determine the
    least $k > m$ for which, when we run $\mathcal{A}$ (which is computable) on
    $G_k$, the nodes in the $m$-th layer of $G_k$ do not see farther than the
    first $k-1$ layers, no matter what input $x$ is chosen for $G_k$.
    By construction, the value computed by $T$ will equal $k_0$.
    \item Construct $G_{k_0}$ with input $x$.
    \item Simulate the operation of $\mathcal{A}$ on $G_{k_0}$ and determine
    what bit $b$ the nodes of $\mathcal{A}$ produce in their output to
    $\Pi_\turing$ in the $m$-th layer.
    \item Output $b$.
  \end{enumerate}
  Note that the first step always terminates due to the existence of $k_0$, and
  so $T$ always halts.

  Finally, let $D$ be the Turing machine from \cref{lem:diagonalization}, $x_T =
  (\descr{D}, \descr{T})$, and $m = \abs{x_T}$.
  Consider the behavior of $\mathcal{A}$ on the instance $G_{k_0}$ with input
  $x_T$.
  Since $T$ always halts, then so does $D$ on input $x_T$.
  Whatever bit the nodes of $\mathcal{A}$ output in the $m$-th layer of
  $G_{k_0}$ is, by construction of $T$, identical to $T(x_T)$.
  However, the output of $U$ on $x_T$ is identical to $D(x_T) = 1 - T(x_T)$.
  Hence $\mathcal{A}$ fails to solve the problem $\Pi$.
\end{proof}

This concludes the proof of \cref{thm:separation}.

\section{Equivalence Results}\label{section:equivalence_results}
In this section we prove our equivalence results. First, we show that once nodes are provided with any global upper bound~$N$ on the network size (and optionally a promise of the form $N \leq f(n)$ for a fixed computable~$f$), the gap between computable and uncomputable \local vanishes: every \ulocal $T$-round algorithm can be converted into a
\clocal algorithm with the same locality.
Second, we show that in uncomputable \local the bound~$N$ is in fact unnecessary: from a family $\{\AA_N\}$ of $T$-round algorithms that work for all $n \leq N$, we construct a single bound-free algorithm that still runs in at most $T(n)$ rounds and agrees with $\AA_N$ for a suitable choice of~$N$. Both arguments are organized around the notion of \emph{maximally $T$-safe neighborhoods} that we describe in the following subsection.

\subsection{Maximally Safe Neighborhoods---and How to Find Them}\label{sec:maximally_safe_neighs}
Here we discuss the concept of maximally safe neighborhood, a useful tool that we will use in our equivalence proofs.
Note that the concept depends only on a given complexity function $T$ and is
\emph{independent} of any particular algorithm under consideration.

\begin{definition}
  Let $T\colon \N_+ \to \N_0$ be a function and $\mathcal{G}$ a class of graphs.
  For a network $(G,x,\ddeg,\id,\{p_v\})$ with $G \in \mathcal{G}$ and a node $v
  \in V(G)$, we say $t \in \N_0$ is \emph{$T$-safe} for $v$ in $\mathcal{G}$ if
  there is no network $(G',x',\ddeg',\id',\{p_v'\})$ with $G' \in \mathcal{G}$
  and $v' \in V(G')$ such that the $t$-neighborhoods rooted at $v$ and $v'$ are
  isomorphic and $t > T(\abs{V(G')})$.
  Moreover, we say $t$ is \emph{maximally $T$-safe} for $v$ if it is maximal
  with this property (i.e., $t$ is safe but $t+1$ is not).
  Without ambiguity, we extend these notions to the rooted neighborhoods
  themselves; that is, if $t$ is (maximally) $T$-safe for $v$, then we say that
  $\NN^t(v)$ is also \emph{(maximally) $T$-safe}.
\end{definition}

Note that, if we fix $v$, then there is always a maximally $T$-safe $t$.
This is because we may always pick $G' = G$, $x' = x$, $\id' = \id$, and $v' =
v$ in the condition above, in which case the only non-trivial requirement is
that $t > T(\abs{V(G)})$, which must hold for all but finitely many $t$.

\begin{lemma}
  \label{lem:max-safe-re}
  Let $T$ be non-decreasing and computable, and let $\mathcal{G}$ be an
  enumerable class of graphs.
  Then, given a rooted neighborhood $X$ of radius $t \in \N_0$, it is
  decidable if $X$ is maximally $T$-safe or not.
\end{lemma}

Since $T$ is non-decreasing, it can be argued that there are only finitely many
counter-examples to test in order to check whether $X$ is maximally $T$-safe
or not.
That being said, we must first deal with some technicalities concerning the
function $T$.

\begin{proof}
  Let $X$ be given.
  Since $T$ is non-decreasing, there are two options for its behavior relative
  to the value $t$:
  \begin{enumerate}[noitemsep]
    \item Either there is $N \in \N_+$ such that $t < T(N)$; or
    \item There is a constant $C \in \N_0$ and a value $N_0 \in \N_+$ for which,
    for every $N' \ge N_0$, $t = T(N') = C$.
  \end{enumerate}
  In the latter case, $C$ is finite information that we can encode in our
  decision procedure, so we can assume it to be known and fixed (together with
  $T$).
  For convenience, we can imagine we have $C = \infty$ in the former case and
  then handle both cases simultaneously.
  Moreover, note that, since $T$ is computable, it is possible in the former
  case to compute a value $N$ with said property (i.e., $t < T(N)$).

  Now notice that, if $t$ is not maximally $T$-safe, then any network with $n$
  nodes that proves this (i.e., containing a rooted neighborhood $X'$
  isomorphic to $X$ that violates $t \le T(n)$) must be such that $n < N$.
  Hence we can execute the following strategy:
  \begin{enumerate}[noitemsep]
    \item If $t = C$, directly accept.
    \item Otherwise, enumerate over every graph $G \in \mathcal{G}$ on $n < N$
    nodes and over every node $v \in V(G)$.
    \item Reject if $X$ is isomorphic to the radius-$t$ neighborhood of $v$ in
    $G$ and $t < T(n)$.
    \item If the enumeration is complete and $X$ has not been rejected, then
    accept.
  \end{enumerate}
  Clearly, this procedure only rejects if it finds a counter-example to $X$
  being maximally $T$-safe.
  Moreover, as we have just reasoned, any such counter-example must come from a
  network on less than $N$ nodes (or, if $t = C$, then it is clear such a
  counter-example cannot exist).
  Hence the procedure correctly decides whether $X$ is maximally $T$-safe.
\end{proof}

\subsection{Equivalence Between Computable and Arbitrary}\label{sec:equivalence_computable_arbitrary}

\begin{theorem}
  \label{thm:equiv-comp-uncomp}
  Let $\Pi$ be an LCL problem, and let $T \colon \N_+ \to \N_0$ and $f \colon \N_+ \to \N \cup \{\infty\}$ be computable non-decreasing functions.
  If there exists an \ulocal algorithm solving $\Pi$ with locality $T$ and knowledge of an upper bound $N$ on the instance size such that $n \le N \le f(n)$, then there exists a \clocal algorithm solving $\Pi$ with locality $T$ and knowledge of an upper bound $N$ with the same assumptions.
\end{theorem}

Given the existence of an \ulocal algorithm solving $\Pi$, we give a \clocal algorithm that also solves $\Pi$ with the same locality.
The idea of the computable algorithm is to construct all possible mappings from maximally $T$-safe neighborhoods to output labelings.
There are only finitely many maximally $T$-safe neighborhoods on at-most-$N$-node-graphs, and only finitely many output labels, so our computable algorithm can enumerate over all mappings from $T$-safe neighborhoods to output labels.
As we know that the uncomputable algorithm exists, we know that this process must find at least one valid mapping.

\begin{proof}
  Let $\mathcal{A}$ be the \ulocal algorithm solving $\Pi$ with the knowledge of an upper bound $N$ on the size of the instance.
  We provide a \clocal algorithm that solves $\Pi$ with the knowledge of an upper bound $N$ on the size of the instance.
  The algorithm works in two phases: it first locally finds a valid mapping from input neighborhoods to outputs, and then applies this mapping.

  The algorithm starts phase one by constructing the set of all valid input
  networks $\mathcal{G}$ such that for every $(G, x, \deg, \id,\{p_v\}) \in \mathcal{G}$
  it holds that $|V(G)| \le N \le f(|V(G)|)$.
  Set $\mathcal{G}$ is finite as we have an upper bound on the size of each network, the input labels come from a finite set, and we have an upper bound on the identifier space.
  The algorithm then collects all maximally $T$-safe neighborhoods $\mathcal{M}$ from these graphs; this is possible by \cref{lem:max-safe-re}.
  Finally, the algorithm enumerates over all possible mappings $g \colon \mathcal{M} \to \Sigma_{\mathrm{out}}$ from these neighborhoods to output labels; this is possible as the set of neighborhoods and the set of possible output labels is finite.
  For each possible mapping, the algorithm tests whether this mapping solves all of the constructed instances and stops once it finds a suitable mapping.
  We defer arguing why such mapping~$g^\star$ exists to the end of the proof.

  Now every node of the network knows the mapping $g^\star$.
  In the second phase, every node finds its maximally $T$-safe neighborhood by enumerating over all radii of the neighborhoods until it finds the maximally $T$-safe neighborhood; this is again possible by~\cref{lem:max-safe-re}.
  Every node then applies mapping $g^\star$ on this neighborhood and halts.
  It is clear that this solves $\Pi$ as $g^\star$ was chosen such that it works on all graphs that satisfy $|V(G)| \le N \le f(|V(G)|)$; in particular the real input instance is such a graph.

  The only thing left to argue is that such function $g^\star$ exists in the first place; if such a function exists, then our algorithm will find it by just enumerating over all of the possible functions.
  Assume now for contradiction that no such function exists.
  Then for every candidate mapping $g$ from maximally $T$-safe neighborhoods to output labels there must exist a network $(G, x, \deg, \id, \{p_v\}) \in \mathcal{G}$ that $g$ fails to label the graph correctly.
  In particular, this holds for the mapping of the \ulocal algorithm~$\mathcal{A}$: either the mapping produces an invalid solution, in which case algorithm~$\mathcal{A}$ doesn't work correctly, or the algorithm produces~$\bot$ to signify that the neighborhood is not large enough.
  Consider this neighborhood~$\mathcal{N}$ for which algorithm~$\mathcal{A}$ produces $\bot$.
  This neighborhood is maximal $T$-safe neighborhood by construction.
  Hence there exists some concrete network $(G', x', \deg', \id', \{p_v'\}) \in \mathcal{G}$ containing $\mathcal{N}$ such that the radius of $\mathcal{N}$ is $T(|G'|)$.
  Consider now executing algorithm~$\mathcal{A}$ on this network: $\mathcal{A}$ requires more than locality~$T(|G'|)$ to label this graph, which is a contradiction on the locality of~$\mathcal{A}$ being $T$.
  Hence we conclude that mapping~$g^\star$ must exist.
\end{proof}

\subsection{Equivalence Between Upper Bound and No Upper Bound on \texorpdfstring{$n$}{n}}\label{sec:equivalence_between-upper-no-upper}

Here we show that, under our assumptions, and  for computable running times, access to an upper bound on the network size does not provide any additional power to deterministic \ulocal algorithms. More precisely, we show that whenever an LCL problem $\Pi$ can be solved by a family of algorithms $\{\AA_N\}$ that know an upper bound $N$ on the number of nodes $|V(G)|$, there also exists an algorithm $\BB$ that solves the same $\Pi$ in the same round complexity without any knowledge of the network size. This establishes an equivalence between the two models: with or without upper bounds, the set of solvable LCL problems is the same.

To this end, fix a computable bound $T:\N_{+}\to\N_{0}$. We know that the set of maximally $T$-safe neighborhoods is enumerable by \cref{lem:max-safe-re}. Let $X_{1},X_{2},X_{3},\cdots$ be a fixed enumeration of all maximally $T$-safe neighborhoods.

\begin{theorem}\label{thm:eliminate-upper-bound}
	Let $\Pi$ be an LCL problem, $\{\AA_{N}\}_{N\in\N_{+}}$ a family of (possibly uncomputable) \local algorithms and $(G,x,\ddeg,\id,\{p_v\})$ a network such that, for every $N\geq n$,
	\begin{enumerate}[noitemsep]
		\item[(i)] the algorithm knows the global upper bound $N$, and
		\item[(ii)] on all graphs with at most $N$ vertices it solves $\Pi$ within $T(n)$ rounds.
	\end{enumerate}
	Then there exists a \local algorithm $\BB$ that
	\begin{enumerate}[noitemsep]
		\item[(a)] has \emph{no information} about $n$, and
		\item[(b)] still solves $\Pi$ on every graph in at most $T(n)$ rounds.
	\end{enumerate}
	Moreover, the output produced by $\BB$ coincides with that of $\AA_{N}$ for some choice of $N$.
\end{theorem}

\begin{proof}
	We split the construction of $\BB$ in three steps and then prove its
	correctness.
  \begin{description}
    \item[Construction step 1: sandboxing the upper-bound algorithms.]
	Based on $\{ \AA_N \}_{N \in \N_+}$, we define a new family of algorithms $\{
	\AA_N' \}_{N \in \N_+}$.
	For each value of $N$, $\AA_{N}'$ does the following on a node $v$:
	\begin{enumerate}
		\item Grow the view until it reaches a maximally $T$-safe neighborhood
		$X\subseteq G$.
		\item Simulate $\AA_{N}$ on $X$.
    \item If the simulation ever inspects a vertex outside $X$, output $\bot$
    and halt.
    \item Otherwise output the label chosen by $\AA_{N}$.
	\end{enumerate}
	We observe that both steps fit in $T(n)$ rounds because $X$ cannot be enlarged without violating the time bound. Write $\AA_{N}'(i)\in\Sigma_{\bot}$ for the value produced when the gathered neighborhood is $X_{i}$. Then, we have:
	\begin{lemma}\label{lem:threshold}
		For every index $i$, there exists $N_{i}\in\N_{+}$ such that $\AA_{N}'(i)\neq\bot$ for all $N\geq N_{i}$.
	\end{lemma}
	\begin{proof}
		Fix $v_i$ and consider its maximally $T$-safe neighborhood $X_i$.
    Let $r_i$ be the radius of $X_i$. Define the graph $G_i = X_i$ and designate its
    root as $v_i$. Observe that the radius-$r_i$ neighborhood of $v_i$ in $G_i$
    equals $X_i$. Set $N_i \coloneqq |V(G_i)| = |X_i|$.
    Then the algorithm $\AA_{N}$ must stop and output a label after $T(N_i)$
    rounds when run on $v_i$ on $G_i$.
    Thus, $\AA_{N}'(i)\neq\bot$.
	\end{proof}
	\cref{lem:threshold} shows that the sandbox never fails, i.e., that for each
	$X_i$ there is a threshold $N_i$ after which $\AA'_{N}(i)\neq\bot$.

	\item[Construction step 2: stabilizing outputs.]
	For this step we inductively define labels
	$s_{1},s_{2},\dots\in\Sigma_{\mathrm{out}}$
	and infinite nested sets
	$S_{0}\supseteq S_{1}\supseteq S_{2}\supseteq\cdots$
	as follows:
	\begin{itemize}
		\item Set $S_{0}\coloneqq\N_{+}$.
		\item To define $s_i$ and $S_i$ for $i>0$ given an infinite set $S_{i-1}$,
		inspect the multiset $\MM_i := \{\AA_{N}'(i)\mid N\in S_{i-1}\}$. By
		\cref{lem:threshold}, at most the first $N_i-1$ indices can map to $\bot$.
		Hence, $\bot$ appears only finitely many times. After discarding those
		finitely many \(\bot\) values, the multiset is still infinite and now
		contains only labels from the finite set \(\Sigma_{\mathrm{out}}\). By the
		infinite pigeonhole principle, at least one output label occurs infinitely
		many times. Let it be \(s_i\). Define $ S_i := \{N\in S_{i-1}\mid
		\AA_{N}'(i)=s_i\}$. Since the class associated with \(s_i\) is infinite,
		the new set \(S_i\) is also infinite.
	\end{itemize}
	Hence, every \(\AA_{N}'\) with \(N\in S_{i}\) returns the same label \(s_{i}\) in the neighborhood \(X_{i}\subseteq G\).

	\item[Construction step 3: the bound-free algorithm $\BB$.]
	We are now ready to define the algorithm $\BB$ that does not use any
	information about the number of nodes in the network. On each node $v$, $\BB$
	does the following:
	\begin{enumerate}
		\item Collect the view until it becomes a maximally \(T\)-safe neighborhood
		\(Y_{v}=X_{i}\);
		\item Output the pre-selected label \(s_{i}\).
	\end{enumerate}
  Clearly, the first step is completed within $T(n)$ rounds, so $\BB$ satisfies
  the time bound.
  \item[Correctness.]  Fix an $n$-node network $(G,x,\ddeg,\id,\{p_v\})$. For
  each node $v\in V(G)$, let $Y_v\cong X_i$ be its maximally $T$-safe
  neighborhood. 
  In addition, let
  \[
	  I = \left\{i \mid \text{$\exists v$  with $Y_v\cong X_i$}\right\}
  \] 
  and $i^\star = \max I$.
  From the stabilization argument (step 2) we have infinite nested sets
    $S_0\supseteq S_1 \supseteq S_2\supseteq \cdots$
  with the property that, for all $N\in S_i$, the sandboxed algorithm $\AA_N'(i)$ outputs the stabilized label $s_i$ on $X_i$. Notice that $S_{i^\star}$ is infinite; choose $N\in S_{i^\star}$, with $N\geq n$. Fix any node $v$ with $Y_v\cong X_i$ (so $i\in I$ and $i\leq i^{\star}$). By nestedness, $S_{i^\star}\subseteq S_i$, hence our chosen $N\in S_i$. Therefore, $\AA_N'(i) = s_i\neq \bot$. By the definition of the sandbox, whenever $\AA_N'(i)  \neq \bot$, it equals the output that the original algorithm $\AA_N$ gives on the root of a node whose $T$-round view is $X_i$. Hence, at every node $v$ we have that algorithm $\BB$ produces  $s_i$ as  $\AA_{N}'(i)$ and  $\AA_N$. Thus the labeling of $\BB$ on $G$ coincides node-by-node with that of $\AA_N$. Since $n\leq N$, algorithm $\AA_N$ is correct on $G$; therefore $\BB$ is also correct.\qedhere
\end{description}
\end{proof}
Consequently, we can claim that in \ulocal knowing an upper bound $N$ on the number of nodes in the graph does not provide any advantage for deterministic algorithms in solving LCL problems, formally:
\begin{corollary}\label{cor:oracle-upper-vs-no-upper}
	In \ulocal, giving all nodes an upper bound $N$ on $n=|V(G)|$ provides no
	advantage for deterministic algorithms solving LCL problems.
\end{corollary}
\begin{proof}
	Suppose an LCL problem $\Pi$ is solvable in \ulocal with access to an upper bound $N$ on the number of nodes  in the considered network $(G,x,\ddeg,\id,\{p_v\})$. Then there exists a family $\{\AA_N\}$ of algorithms as in \cref{thm:eliminate-upper-bound} that solve $\Pi$ in $T(n)$ rounds. By the theorem, there exists a bound-free algorithm $\BB$ that solves $\Pi$ in the same number of rounds as well.
\end{proof}

\ifanon\else
\section*{Acknowledgments}

We thank Francesco d'Amore for posing a question that prompted us to embark on
this project.
This work was supported in part by the Research Council of Finland, Grants
359104 and 363558.
Diep Luong-Le was supported by the Aalto Science Institute International
Summer Research Programme.
Most of this work was done while Augusto Modanese was affiliated with Aalto
University.

\fi 

\printbibliography
\end{document}